\documentclass[letterpaper,boxed,12pt]{amsart}
\usepackage[utf8]{inputenc}
\usepackage[margin=1in]{geometry}
\usepackage{amsmath,pgfplots}
\usepackage{enumerate}
\usepackage{amssymb,amsthm,tikz, ulem}

\theoremstyle{plain}
\newtheorem{theorem}{Theorem}[section] % reset theorem numbering for each chapter

\theoremstyle{definition}
\newtheorem{defi}[theorem]{Definition} % definition numbers are dependent on theorem numbers
\newtheorem{exmp}[theorem]{Example}

\newtheorem{lemma}[theorem]{Lemma}
\newtheorem{prop}[theorem]{Proposition}
\newtheorem{remark}[theorem]{Remark}

\newcommand\RR{\mathbb{R}}
\newcommand\NN{\mathbb{N}}
\newcommand\cT{\mathcal{T}} % tree space with leaf edges
\newcommand\cO{\mathcal{O}} % orthants
\newcommand\leafset{\mathcal{L}}    % leaf sets, and leaves of T are \leafset(T)
\newcommand\cE{\mathcal{E}} % edge set of a tree; need to use mathcal to distinguish from extension set

\newcommand{\BHV}{\textrm{BHV}}

\title{Geometric comparison of phylogenetic trees with different leaf sets}
\author{Gillian Grindstaff}
\address{Department of Mathematics, The University of Texas at Austin}
\email{gillian.grindstaff@math.utexas.edu}

\author{Megan Owen}
\address{Department of Mathematics, Lehman College, City University of New York}
\email{megan.owen@lehman.cuny.edu}
\date{}

\begin{document}
\maketitle

\begin{abstract}
The metric space of phylogenetic trees defined by Billera, Holmes,
 and Vogtmann~\cite{BHV}, which we refer to as $\BHV$ space, provides a natural geometric setting for
 describing collections of trees on the same set of taxa.  However, it
 is sometimes necessary to analyze collections of trees on
 non-identical taxa sets (i.e., with different numbers of leaves), and
 in this context it is not evident how to apply $\BHV$ space.
 Davidson et al.~\cite{UIUC} approach this problem by describing a
 combinatorial algorithm extending tree topologies to regions in
 higher dimensional tree spaces, so that one can quickly compute
 which topologies contain a given tree as partial data.  In this
 paper, we refine and adapt their algorithm to work for metric trees
 to give a full characterization of the subspace of extensions of a
 subtree.  We describe how to apply our algorithm to define and search a space of possible supertrees and, for a collection of tree fragments with different leaf sets, to
 measure their compatibility.
\end{abstract}

\section{Introduction}
In the context of evolutionary biology, given a set of organisms referred to as taxa, a phylogenetic tree is a semi-labeled, weighted acyclic graph representing a possible evolutionary relationship between the taxa, using genotypic or phenotypic data. Such trees typically have a root which represents the common ancestor of the taxa, with a branch point at each speciation event, and a leaf for each taxon, such that the taxa which share more features are ``nearer" to each other in the tree. Here the intrinsic tree distance is exhibited by shortest path length in the weighted tree: a series of edges without repetition gives a unique path from one leaf to another, and the sum of their lengths is distance, indicative of the genetic or phenotypic changes and differences between the taxa. 

In topological data analysis, phylogenetic trees represent an important class of metric spaces, finite additive spaces, which exhibit no persistent topological features at any scale in degree larger than 0. In this way, techniques such as these we present can be viewed as complementary to homological methods.
\begin{figure}
\label{Tree of Life}
\begin{center}
\includegraphics[scale=.7]{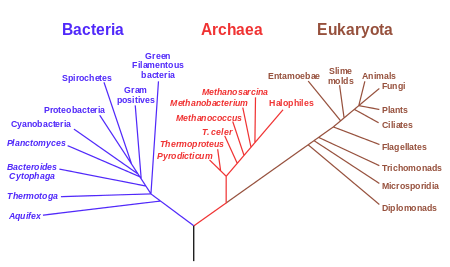}\caption{Phylogenetic Tree of Life \cite{Wiki}}
\end{center}
\end{figure}

%Phylogenetic trees also have many applications outside evolutionary biology and bioinformatics, including anatomy, comparative linguistics, personalized medicine, cultural evolution, and literary studies. 

In addition to the intrinsic distances between the taxa that a single phylogenetic tree represents, we can also define an extrinsic distance between distinct phylogenetic trees with the same set of taxa. In 2001, Billera, Holmes, and Vogtmann defined a configuration space of possible phylogenetic trees relating a set of taxa \cite{BHV}. These trees can be continuously parametrized by the topology and edge lengths, and the result is a contractible geodesic space with non-positive curvature, referred to here as tree space, or $\mathcal{T}^n$, for taxa labels $\{1,\dots,n\}$. The extrinsic distance between two trees is realized as the length of the unique geodesic between the two points in $\mathcal{T}^n$ which represent the trees.
%
%Though there are many sophisticated algorithms for inferring a tree from sampled data, because of a variety of factors, different data samples often yield different shapes (``topologies") of trees. A natural problem, then, is to aggregate these tree estimates, and represent a collection of phylogenies for a fixed taxa set in some meaningful way for analysis. 

 To combine the data of more than two trees, e.g. if $\mathbf{T} = \{T_i\}$ is a set of phylogenetic trees describing different evolutionary relationships between the taxa (leaf set) $\leafset$, $\mathbf{T}$ is represented as a set of points in $\mathcal{T}^n$. By taking the mean of $\mathbf{T}$ \cite{mop, bacak,bacak_applied,brownowen}, or clustering the points \cite{gori_etal}, or constructing confidence regions \cite{willis}, we can describe $\mathbf{T}$ in a way which incorporates the range of metric and combinatorial shape differences. 

However, there are situations in which one of the assumptions of this model, that each tree in $\mathbf{T}$ has a fixed leaf set $\leafset$, is not reasonable.  For example, with improvements in sequencing technology, many phylogenetic datasets now consist of thousands of gene trees, each of which represents the evolutionary history of a single gene in the species set of interest \cite{Maddison}.  However, not all genes appear in all species, and currently genes with an incomplete leaf-set are often discarded before beginning the analysis.  A second example is comparing parallel evolutionary chains in viruses or tumors, where some strains are comparably similar across samples (and therefore can be considered the same leaf) but are not necessarily all present in every sample \cite{ZKBR2}, i.e. each $T_i\in \mathbf{T}$ has its own leaf set $\leafset_i$ which is contained in some common larger set $[N]$. The fact that the trees $T_i$ belong to different parametrized spaces prevents us from using the techniques of BHV analysis described previously, but as we will show, tree sets with some ``combinatorial compatibility" will admit a fairly precise notion of distance which is based on the BHV metric in $\mathcal{T}^N$, with no loss of data. 

%Another application is related to 
Our approach to this problem uses the {\bf tree dimensionality reduction} map $\Psi$ defined in Zairis et al. \cite{ZKBR2}, which gives a map from a tree space $\mathcal{T}^N$ to the lower-dimensional tree space $\mathcal{T}^{\leafset}$ that contains all trees with a subset of the leaves $\leafset \subset[N]$.  This map is induced by the natural subspace projection. % $([N],d) \mapsto (\leafset,d)$. 
We will first construct the pre-image $\Psi^{-1}$ of this map, which can be used to recover information about the original tree $T$ from the images $\{\Psi_{\leafset}(T)\}$ for varying $\leafset$. This map $\Psi$ is also fundamental to the previous applications, which we solve by mapping $T_i$ to their preimages $\Psi^{-1}(T_i)$ in the common domain space $\mathcal{T}^N$, and comparing the sets.

This precise problem, of analyzing trees with different numbers of taxa collectively in $\BHV$ tree space, was first approached by Bi et al. \cite{UIUC}. They developed the theory behind the combinatorial step in Section \ref{extcombstep}, toward the goal of comparing trees with different taxa sets. The algorithm presented in that section, together with Proposition \ref{CT}, clarifies their results and shows their implications for the computation of tree dimensionality reduction and its preimage.

Analysis in $\BHV$ space is, of course, not the only way to approach problems of this type. Given the set $\{T_i\}$, it is sometimes efficient to ``prune" the trees to their common taxa $\cap_i \leafset_i$ for comparison, %using $\Psi$, 
if such a set $\cap \leafset_i$ is sufficiently large to preserve important data. In this case, any tool for analyzing sets of trees with identical taxa can then be used.  In the context of reconstructing a species tree from gene trees, the relationship between these trees is modeled by the coalescent process, and algorithms and approaches can specific to this situation can take advantage of this model \cite{Rhodes, Astral2}. To avoid making simplifying assumptions, there are also some software packages currently available which use Bayesian coalescent-based techniques, from the original data rather than trees, to assemble multiple parallel, incomplete data samples into a single tree \cite{BEAST, starBeast, BEST}. There are also algorithms, based on the (often reasonable) assumption that differences in topology arise from recombination events, that aggregate metric data into phylogenetic networks \cite{PhyloNet}. These can often accommodate non-uniform data as well. However, they share the same drawback as most classical phylogenetic tree algorithms, in that they produce a single tree or tree-like object, rather than a region of possible trees in tree space.

There is also the problem of supertree reconstruction, which aims to combine partially overlapping phylogenies into a common tree. Summaries and selected supertree methods can be found in Bininda-Emonds \cite{BE}, Akanni et al. \cite{RSOS}, Warnow \cite{warnow}, and Wilkinson et al. \cite{SB}. The techniques in this paper give a conservative (low tolerance for topological error), split-based supertree method for BHV space, which does not necessarily represent an improvement on the search for a maximum-likelihood supertree; rather, we can rigorously (rather than heuristically) define the space of possible supertrees, in a manner amenable to search, and expand the possible analyses available. 

 With the geometric framework established in this paper, we can define and compute some useful objects. First, in Section 3, we show how to efficiently compute $\Psi^{-1}(T)$, the preimage of tree $T$ under the tree reduction map, which gives all trees with the full set of leaves $N$ that map onto $T$. The algorithm, given in two parts, calculates the extension space $E_{T,n,\ell}$, which represents the set of all phylogenetic trees which can result from adding $\ell$ additional leaves to $T\in \mathcal{T}^n$. Theorem \ref{Eispreimage} shows that this construction,	which extends the results and definitions of \cite{UIUC}, coincides with $\Psi_n^{-1}(T)$ in $\mathcal{T}^{n+\ell}$.

This fact immediately gives a method of finding the set of trees $X$ which satisfy the system $\{\Psi_{\leafset_i}(X) = T_i\}$ for some collection of trees $\mathbf{T} = \{T_i\}$, and we suggest some shortcuts to speed up the process. This solution space $E_{\mathbf{T}}$ is computed efficiently in Section 4 in a method similar to the one presented in Section 3, and is shown in Proposition \ref{intext} to be the intersection of sets $\Psi_{\leafset_i}^{-1}(T_i)$ in a common domain.

Stability concerns lead us to Section 5, which first defines an approximate solution space to $\{\Psi_{\leafset_i}(X) = T_i\}$ with some parameter $\alpha$ of constant error tolerance, or $p_\alpha$ of error tolerance proportional to local size. These will be the products of Sections 5.1 and 5.2, and will allow for stability results (\ref{neighborhood}) and (\ref{alphastab}). The proposition (\ref{neighborhood}) implies an additional non-trivial fact about a set $\Psi^{-1}(T)$, that if it intersects a cubical face $\sigma \subset \mathcal{T}^N$, it intersects all cubes $\tau\supset \sigma$.

From these definitions we can define two parameters $\alpha_\mathbf{T}$ and $p_\mathbf{T}$ measuring the degree of metric distortion for a collection of trees $\mathbf{T} = \{T_i\}$ satisfying a combinatorial compatibility condition. The parameters represent the minimum error tolerance (uniform or proportional) necessary to construct a supertree from the $\{T_i\}$. These parameters will result from linear optimization problems related to the equations defining the approximate solutions spaces, and can be directly computed using the most efficient linear programming methods available.

% Additionally, the proportional parameter $p_\mathbf{T}$ can be utilized as a ``test" of compatibility, using some threshhold value for error tolerance, which is explained in Section \ref{Hyp}.

% In Section \ref{rot}, we take a first step toward extending these tree set parameters past the case of combinatorial compatibility with the definition of a rotation relaxation parameter, which can be combined with the metric parameters to give a similar threshholding test in Section \ref{rothyp}.

Some directions for future work are sketched in Section \ref{FW}. The suggested projects include giving a full extension of the definitions, computations, and parameters to tree sets which need not be combinatorially compatible, and relaxations which can exceed the boundaries of its supporting orthants. Additionally, a probabilistic framework for random trees on different leaf sets would be able to give significance to the threshholding tests, which in this paper are merely heuristic.

\section{Background}

\subsection{Phylogenetic trees}
\begin{defi} A {\bf phylogenetic tree} $T$ is an acyclic connected graph (a {\bf tree}) with
\begin{itemize}
\item No degree 2 vertices.
\item Degree 1 vertices each have a unique label. Such vertices are called {\bf leaves} of $T$. The set of leaf labels is denoted $\leafset(T)$.
\item There is a positive weight $w_e$ for each edge $e$, and the set of edges is denoted $\cE(T)$. 
\end{itemize}
% note that E(T) is the set of all edges, leaves and internal
Unless indicated otherwise, $\leafset(T) = [n]=\{1,2,\dots, n\}$ for $n$ the number of leaves. Phylogenetic trees are sometimes {\bf rooted}, meaning the tree has a distinguished leaf, the {\bf root}, often an ancestor. For this paper, we will use unrooted trees, but all results carry over to rooted trees by fixing one of the leaves as the root, and assuming this leaf is in all trees considered.  The {\bf topology} of a tree is the unweighted underlying tree with leaf labels. 
\end{defi}
Because phylogenetic trees are acyclic, the removal of an edge $e$ separates $T$ into two connected components. Since leaves are vertices in one component or another, this gives a partition of $\leafset(T)$ into the two components $P_e$ and $P_e^c = \leafset(T)\setminus P_e$, called a {\bf split} and represented as $P_e | P_e^c$. When the ground set is obvious, we will suppress the complement and give a split by the smaller of its two partition sets, or if the two partitions are the same size, with the partition containing the lexicographically first leaf.  A split is called {\bf thick} if $P_e$ and $P_e^c$ both have cardinality greater than 1, or equivalently if neither endpoint of $e$ is a leaf. Such an edge or split is called {\bf internal}.

\begin{defi}\label{compatible} Two splits $P$ and $Q$ are called {\bf compatible} if one of: $P\cap Q, P\cap Q^c, P^c \cap Q, P^c \cap Q^c$ is empty. Two splits that are not compatible are called {\bf incompatible}.
\end{defi} It is easy to see that one of these intersections being empty implies that the other three are non-empty. Compatibility of different splits $P$ and $Q$ is equivalent to the existence of a tree $T$ such that the removal of one edge of $T$ gives $P$, and the removal of another gives $Q$.

In fact there is a deep duality between phylogenetic trees and split sets: given a set of $i$ different splits on leaf set $\leafset$ which are pairwise compatible, and weights for each, there is a unique phylogenetic tree realizing them (Buneman et al., 1971 \cite{BHKT}). Conversely, for a phylogenetic tree $T$, the collection of all splits $S(T) = \{P_e\}$ (one for each internal edge $e$) is pairwise compatible. A phylogenetic tree contains at most $2|\leafset(T)| - 3$ splits, and $|\leafset(T)| - 3$ thick splits.  It will be very useful for us to have this structural equivalence between a phylogenetic tree $T$, and the split set $S(T)$ which defines its topology. We will alternately refer to an edge $e\in T$ and the partition $P_e$ it induces; for both, the weight is denoted $w_e$.

If the external (leaf) edges of $T$ are also endowed with weights, then $T$ is equivalent to an additive metric space, whose points are leaves with the weighted path metric on $T$. This is discussed further in Section 2.4.

\subsection{Tree Space}
For a fixed leaf set $\leafset$ and a set of compatible thick splits $S$ on $\leafset$, there exists a unique tree topology realizing $S$, as discussed in the previous section. We can then organize the set of all phylogenetic trees with this topology by their weight sets, ordered lexicographically by the corresponding split of each weight, in a space isometric to $\mathbb{R}_+^{|S|}$.  We can include the boundary, by allowing weights to be 0, and this gives us $\RR_{\geq 0}^{|S|}$, which is called an {\bf orthant}.  Maximal orthants have dimension $|\leafset|-3$.  %In this space, a tree $T$ with splits $S$ is represented by the non-negative vector $(w_1,w_2,\dots, w_{i})$ of split weights, $w_j$ the length of $S_j$ in $T$. 
This is illustrated in Figure~\ref{fig:tree_1}. 
The {\bf norm} of a tree $||T||$ is the $L^2$ norm of the vector of its split weights.  We will denote the lowest-dimensional orthant containing tree $T$ by $\cO(T)$, and the lowest-dimensional orthant containing all trees with exactly the splits $S$ by $\cO(S)$.  Conversely, the set of splits contained in all trees in the interior of orthant $\cO$ is denoted by $S(\cO)$.

\begin{figure}
    \centering
    \includegraphics[scale=0.4]{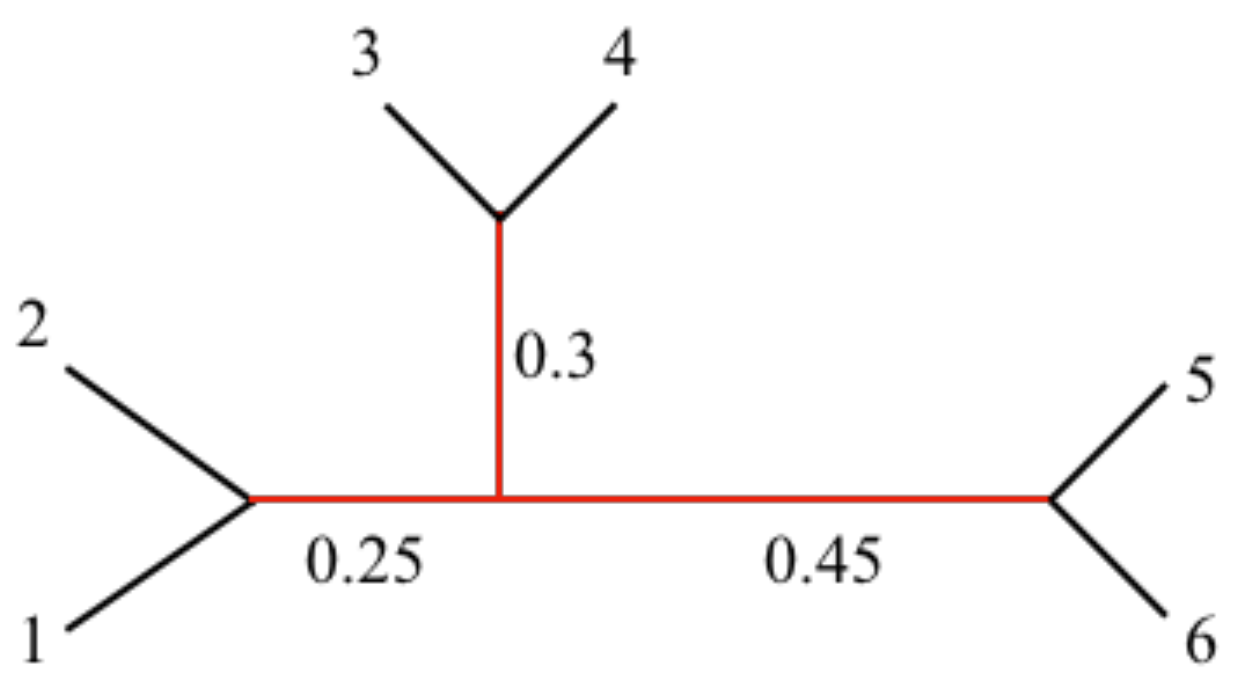}
    \includegraphics[scale=0.4]{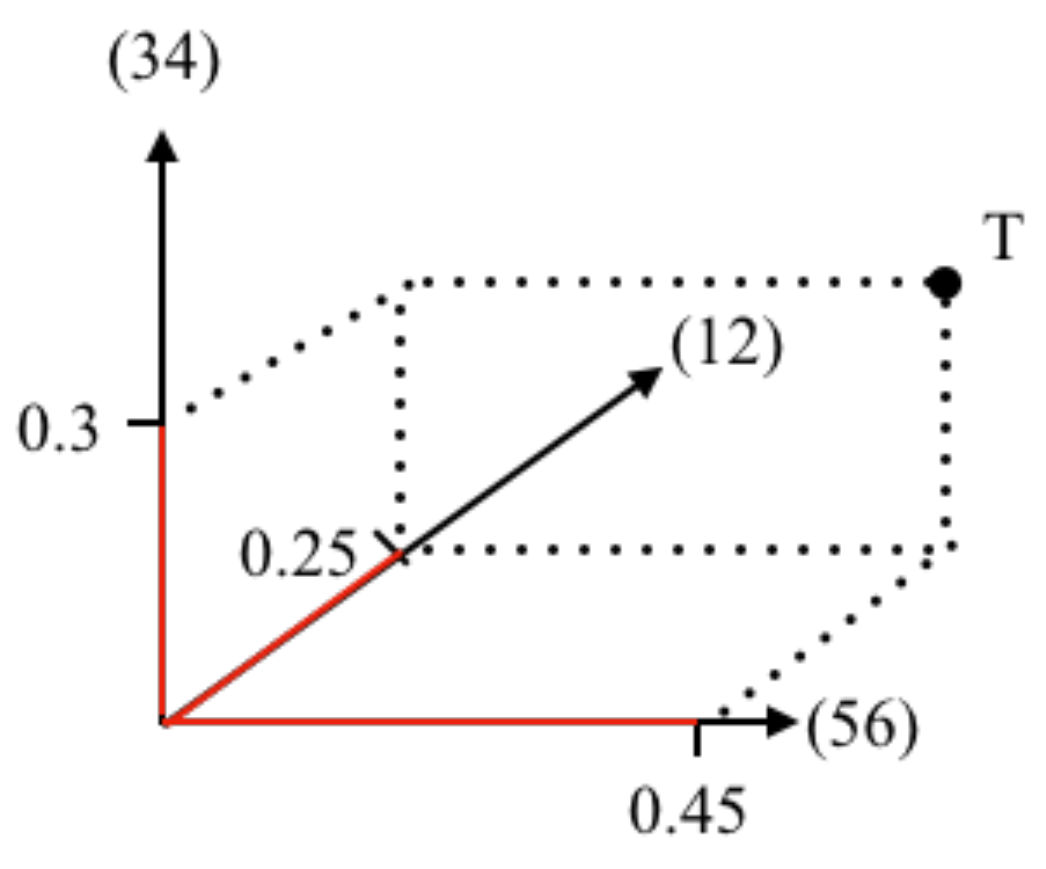}
    \caption{Left, a tree $T$ with 6 leaves, splits $(01)(2345)$, $(23)(0145)$, and $(45)(0123)$ with weights $0.25$, $0.3$, and $0.45$, respectively. Right, the point in the $\{(01),(23),(45)\}$ orthant of $\mathcal{T}^6$ representing $T$. The cone point $\mathbf{0}$ is shown at the origin.}
    \label{fig:tree_1}
\end{figure}

If two sets of compatible thick splits, $S_1$ and $S_2$, have splits in common, $C = S_1 \cap S_2$, then the orthants corresponding to $S_1$ and $S_2$ each have a boundary orthant $\RR^{|C|}_{\geq 0}$ that contains the same trees.  We identify all such common boundary orthants to produce a single space, called the {\bf Billera-Holmes-Vogtmann (BHV) treespace} and denoted $BHV_{\leafset}$, where $\leafset$ is the leaf-set of all trees.  When $\leafset = [n]$, we will alternatively write $BHV_n$ for the space. The empty split set $S = \emptyset$ produces a single point, called the {\bf cone point}, $\mathbf{0}$, which represents the unique star-shaped tree with no internal edges. The cone point is contained in each orthant at the origin, so the identified space is path-connected. We define the distance $d_{\BHV}(T,T')$ between points $T$ and $T'$ in this space to be the infimum of the lengths of all piecewise smooth paths from $T$ to $T'$, where path length is calculated by summing the $L^2$ distances of the path restricted to each orthant it passes through. The norm $||T|| = d_{\BHV}(\mathbf{0},T)$, via the straight line path from the origin to the $T$ point in the orthant containing $T$. 

The BHV treespace was first proposed by Billera, Holmes, and Vogtmann in \cite{BHV}, where they showed that it is a contractible, complete, and globally non-positively curved, or CAT(0), cube complex. Global non-positive curvature implies that there is a unique shortest path, or geodesic, between each pair of trees in the space.  %It is also geodesic, meaning the metric $d_{\BHV}$ as an infimum over paths is achieved by a (piecewise linear) path from $T$ to $T'$. 
There exists a polynomial time algorithm to calculate this path and its length, given by Owen and Provan in \cite{OP}.

For the purposes of this paper, we will have to keep track of the weights of edges ending in leaves as well, but since all trees in $\BHV_{\leafset}$ have the same leaves, and therefore the same leaf partitions, we can represent these globally with non-negative coordinates $(\mathbb{R}_{\geq 0})^{|\leafset|}$, and define {\bf tree space} $\mathcal{T}^{\leafset}$ with this product
$$ \mathcal{T}^{\leafset} := \BHV_{\leafset}\times (\mathbb{R}_{\geq 0})^{|\leafset|}$$
%Then along the boundary of $\mathcal{T}^n$ there are trees which have some leaves identified, and the cone point is again the origin in each coordinate, 
In this case, the cone point is the tree with no edges and all leaves identified into a single point. Importantly, $\mathcal{T}^{\leafset}$ has all of the important features of $\BHV_{\leafset}$: it remains connected, globally non-positively curved, and contractible.  As above, when $\leafset = [n]$, we may alternatively write $\cT^n$ for the space.  The distance $d_{\cT^{\leafset}}(T,T')$ between trees $T,T' \in \cT^{\leafset}$ can also be computed by a version of the algorithm of Owen and Provan \cite{OP}.
%, and geodesics are componentwise geodesic, i.e. $\gamma:I \to \BHV_n$ and $\tau: I \to (\mathbb{R}^+)^n$ are both geodesic iff $(\gamma,\tau):I \to \mathcal{T}^n$ is geodesic. 

\subsection{Link graph}
\begin{defi} The {\bf link} $L_{\leafset} := L_{\leafset}(\textbf{0})$ of the cone point \textbf{0} is the set of all trees in $BHV_{\leafset}$ which have internal edge lengths summing to 1.  Homeomorphically, $L_{\leafset}$ is the set of trees in $BHV_{\leafset}$ at fixed $L_1$ distance from \textbf{0}. 
\end{defi}

Because $\BHV_{\leafset}$ is a cube complex, $L_{\leafset}$ is a simplicial complex; the face maps are restrictions of face maps of the cube complex, and every $k$-face of the cube complex intersects the link in a $(k-1)$-simplex. In particular, the 0-simplices correspond to splits of length 1, the 1-simplices correspond to compatible split pairs, and $k$-simplices correspond to trees sharing the same $k$ non-zero splits which have edge lengths summing to 1. 

$\BHV_{\leafset}$ can then be expressed as a cone on $L_{\leafset}$ based at \textbf{0} (hence the name ``cone point"), with the cone dimension parametrizing magnitude. Denote the 1-skeleton of the link $L_{\leafset}^1$. The global non-positive curvature condition on $\BHV_{\leafset}$ gives that $L_{\leafset}$ is a flag complex, meaning that each $k$-clique in $L_{\leafset}^1$ bounds a $k$-simplex in $L_{\leafset}$, which corresponds uniquely to the orthant of dimension $k$ spanned by the $k$ splits. This means that $L_{\leafset}$ is recoverable from $L_{\leafset}^1$, which together encode all of the non-linearity of $\BHV_{\leafset}$. In \cite{UIUC}, and in the algorithm presented in section 3.3, $L_{\leafset}^1$ is used to calculate the (combinatorial) extension objects $G_{T_s,n,\ell}$ and $S_{T_s,n,\ell}$.

\subsection{Tree dimensionality reduction}
A weighted graph, endowed with the shortest path metric, is a metric space whose underlying set is the vertices of the graph. Acyclic graphs have unique geodesics, and so a metric tree with $n$ leaves can be equivalently considered as a metric on the set of $n$ leaves, with distance between two leaves given by the length of the unique path between them. A metric $\delta$ which arises from a tree in this way is called an {\bf additive} metric, and satisfies the four point condition:
$$\delta(a,b) + \delta(c,d) \leq \max\{\delta(a,c)+\delta(b,d),\delta(a,d) + \delta(b,c)\}$$
for all leaves $a,b,c,d$.

The four point condition is also sufficient to determine additivity, which in turn implies the existence of a unique tree realizing this metric \cite{BHKT}.  The {\bf additive distance matrix} of a tree $T$ with leaf-set $\leafset = \{\ell_1, \ell_2, ..., \ell_n\}$ is denoted $A_T$ and is an $n \times n$ matrix where the $(i,j)$-th entry is $\delta(\ell_i, \ell_j)$, the distance between leaves $\ell_i$ and $\ell_j$ in tree $T$.

A subspace of an additive metric space is additive, and additive subspaces can be seen as forming subtrees.  {\bf Tree dimensionality reduction (TDR)}, as defined in \cite{ZKBR2}, is a method of generating the tree for a subspace of an additive metric space from the original metric tree, and for a more general class of metric spaces called ``nearly" additive. In this paper we will deal exclusively with additive metric spaces. 

\begin{defi}\label{TDR} Let $T$ be a tree with leaf set $[N] = \{1,2,\dots,N\}$, and let $\leafset \subset [N]$.  The {\bf tree dimensionality reduction map} $\Psi_{\leafset}:\mathcal{T}^{[N]}\rightarrow \mathcal{T}^{\leafset}$ is the map sending $T \in \mathcal{T}^N$ to the induced subtree spanned by the leaves $\leafset$
%$[N]\setminus L$
, where the induced subtree contains the vertices and edges on the shortest paths through $T$ between the leaves in $\leafset$, with each resulting degree 2 vertex $v$ and its incident edges $(v,u_1),(v,u_2)$ with lengths $\ell_1$ and $\ell_2$ respectively, being replaced by a single edge $(u_1,u_2)$ with length $\ell_1+\ell_2$. We refer to this process as {\bf concatenation} of $(v,u_1)$ and $(v,u_2)$. 
\end{defi}

We will also consider just the combinatorial reductions of splits, which we will refer to as {\bf projections}, and which simply remove some of the leaves from one or both partitions of a split.  For a split $P|P^c$ on leafset $[N]$, the projection onto the leaf-set $\leafset \subset [N]$ is the split $(P \cap \leafset)(P^c \cap \leafset)$.  Note that one of $(P \cap \leafset)$ or $(P^c \cap \leafset)$ may be empty, in which case we would then discard this split.  Since the tree dimensionality map $\Psi_{\leafset}$ operating on tree $T \in \cT^N$ has the effect of projecting all splits $S =S(T)$ onto the leaf-set $\leafset$, we will abuse notation and use $\Psi_{\leafset}(S)$ to represent this combinatorial projection. 
% and leaves $ and leaf $g \in \leafset$, $\Psi_{\bar{g}}(P)$ will be the split $(P\setminus\{g\})(P^c)$ if $g \in P$, and $(P)(P^c \setminus \{g\})$ if not. Thus a $g$-pruning of a tree induces a combinatorial pruning of each of its splits. %When not indicated otherwise, $\Psi_{n}$ will be the projection onto the first $n$ labels of a metric tree (with edge concatenation), a combinatorial tree (without edge concatenation), or a split. In this way $\Psi$ is a dimension reduction on $\mathcal{T}^{n+1}$, $\BHV_{n+1}$ and $L_{n+1}^1$ to the $n$-dimensional versions of each.  \megan{fix notation here... I'm not sure $\Psi$ should operate on both trees and splits} \gill{$\Psi$ acts in a related way on both, which I used pretty heavily in the next chapter; I wouldn't mind altering the notation but it should still be very similar. $\Psi^*$ or something?}

The following result states that the dimensionality reduction should also give the tree which will be constructed from partial information.

\begin{prop}[{{\cite[Proposition~4.4]{ZKBR2}}}]
\label{prop:subset_additive_matrix}Let $T$ be a tree with leaf set $[N] = \{1,2,\dots,N\}$, and additive distance matrix $A_T$. Let $\leafset \subset [N]$, and define $(A_T)_{\leafset}$ to be the submatrix of $A_T$ with rows and columns indexed by $\leafset$. Then $A_{\Psi_{\leafset}(T)} = (A_T)_{\leafset}$.
\end{prop}

Note that this formulation implies that certain dimension reductions act like projections: if $\leafset \subset \leafset' \subset [N]$, then $\Psi_{\leafset}\circ\Psi_{\leafset'} = \Psi_{\leafset}$ on $\mathcal{T}^{N}$.   

\begin{exmp} Starting with the tree on the left in Figure~\ref{fig:trees_2}, tree dimensionality reduction to the leaf set $\{1,2,3,4\}$ is performed by first pruning the 5th leaf and its leaf edge, which gives the center tree. This tree has a degree 2 vertex, in red, which is removed, its boundary edges concatenated, to produce the final tree on the right. 

\begin{figure}
    \centering
    \includegraphics[scale=0.5]{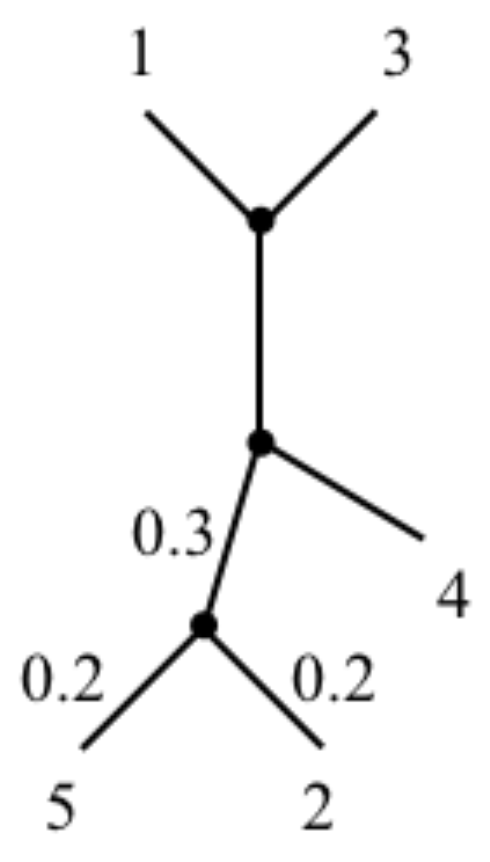} \quad\quad\quad
    \includegraphics[scale=0.5]{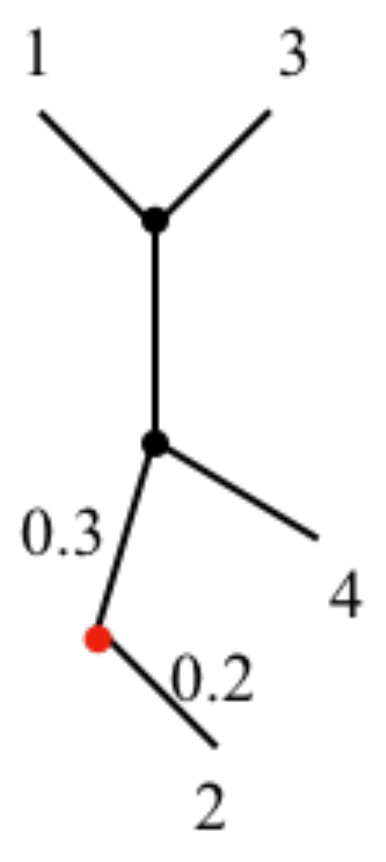} \quad\quad\quad
     \includegraphics[scale=0.5]{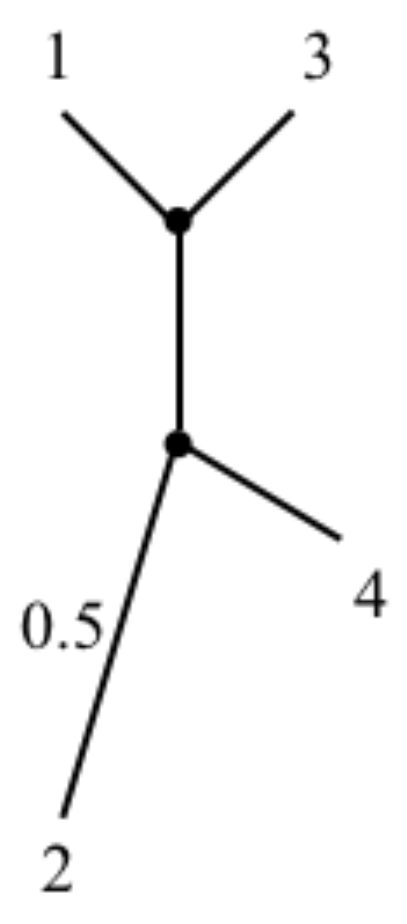}
    \caption{Left, a tree with 5 leaves. Center, the tree with leaf 5 and its edge deleted, resulting in a degree two vertex (in red).  Right, the tree after concatenating the two edges adjacent to the degree two vertex.}
    \label{fig:trees_2}
\end{figure}

\end{exmp}

\section{The Pre-Image of the Tree Dimensionality Reduction Map}
%\section{Extension of $n$-tree to $\mathcal{T}^{n+\ell}$}
The aim of this section will be to algorithmically construct the preimage of the tree dimensionality reduction map $\Psi_{\leafset}:\mathcal{T}^N\to \mathcal{T}^{\leafset}$, for $\leafset \subset [N]$, $|\leafset| = n$. We start with a binary tree $T \in \mathcal{T}^{\leafset}$ with %leaf labels $L(T)=\{1,2,\dots,n\}$ and 
edge lengths $w_e$ for $e\in \cE(T)$, and want to describe and compute the set of all trees $\bar{T} \in \cT^N$ such that $\Psi_{\leafset}(\bar{T}) = T$.
%
%. Then for $N = n+\ell$, the TDR map $\Psi_n: \mathcal{T}^N\to \mathcal{T}^n$ sends trees $\bar{T}\in \mathcal{T}^N$ with leaf labels $\{1,2,\dots,n,n+1,\dots,n+\ell = N\}$ to the metric subspace spanned by the first $n$ leaves, as prescribed by Definition \ref{TDR}. 
%
%Our goal, then, will be to define and compute the set of all $\bar{T} \in \mathcal{T}^N$ such that $\Psi_n(\bar{T}) =T$. 
Since by Proposition~\ref{prop:subset_additive_matrix}
the distance of the leaves $N \backslash \leafset$ to each other and to the leaves $\leafset$ does not affect the distance between the leaves $\leafset$, 
%$n+1,\dots, n+\ell$ to each other and to the first $n$ leaves does not affect the subspace spanned by $\{1,2,\dots,n\}$, 
many different tree topologies can map to $T$ under $\Psi_{\leafset}$.  Thus it is not immediately obvious how this set $\Psi_n^{-1}(T)$ should be described.

As this section demonstrates, one effective approach is to:
\begin{enumerate}
\item Note that for any $\bar{T} \in \cT^N$, the topology of the image $\Psi_{\leafset}(\bar{T})$ is completely determined by the topology of $\bar{T}$, and $\Psi_{\leafset}$ acts linearly on the $\cE(\bar{T})$ edge weights in the orthant $\cO(\bar{T})$ in $\mathcal{T}^N$.  Thus, for a fixed maximal orthant of $\mathcal{T}^N$, $\Psi_{\leafset}$ restricts to a linear map $M:\mathbb{R}^{2N-3} \to \mathbb{R}^{2n-3}$.  Any non-maximal orthant is on the boundary of at least three maximal orthants, and the linear map of any of these maximal orthants can be used.

\item Find the orthants with a topology $\bar{T}$ such that $\Psi_{\leafset}(\bar{T})$ has the same topology as $T$. By Proposition \ref{CT}, these orthants can be determined by individual and pairwise properties of their splits, a surprising result.

\item For a fixed orthant $\cO$, form the matrix $M^{\cO}$ which encodes the way the edges of trees in $\cO$ concatenate under $\Psi_{\leafset}$.

\item Find the positive solutions of the linear system of equations $M^{\cO}\mathbf{x}^{\cO}=\mathbf{w}$, where $\mathbf{w}$ is the vector of edge weights in $T$, to determine the points $\bar{T}\sim\mathbf{x}^{\cO}\in \cO$ such that when $\Psi_{\leafset}$ is performed, all of the edges of $\bar{T}\in \cO$ which concatenate to form an edge $e\in T$ have weights summing to $w_e$.

\item Take the union of all of the orthant-wise solutions, and call this the {\bf extension space} $E^N_T$.
\end{enumerate}

We will show that $E^N_T = \Psi_{\leafset}^{-1}(T) \subset \mathcal{T}^N$, and that the resulting space is connected, continuous, piecewise linear, of local dimension $2(N-n)$, and computable in cubic time relative to its size.  We call the above algorithm the {\bf extension algorithm}.

Note that we will assume that $T$ is binary, since an unresolved tree is often used in biology when the underlying relationship of certain leaves or subtrees is not known.  In such cases, the edge lengths near the unresolved vertex would not necessarily represent the expected length of their corresponding split in the true tree, which is the main assumption of this paper.  Thus we focus on binary trees in this paper, and leave incorportating unresolved trees into this framework for future work.

%%%%
%   Metric connection step 1 
%%%%
\subsection{Extension by one leaf}

To give some intuition for how the extension space relates to the original tree, and to show the mechanics of the base case for later results, we first examine the case where $N = |\leafset|+1$. This means finding the set of trees $\Psi_{\leafset}^{-1}(T)$ which have one additional leaf, labeled $g$.  % changed the leaf to be $g$ so it doesn't conflict with $ell$ = # of new leaves, as used in the rest of the paper

\begin{defi} Let $\Psi_{\bar{\ell}}:\mathcal{T}^N \to \mathcal{T}^{N \backslash g}$ be the tree dimensionality reduction map which deletes leaf $g \in [N]$ and its adjacent edge, and concatenates the two edges at leaf $g$'s attachment point.  We will refer to this reduction as an $g$-{\bf pruning}.
\end{defi}

The reverse of pruning a leaf $g$ is attaching a new leaf $g$ to the tree with a new edge.  We call this attachment operation grafting.

\begin{defi} For a tree $T \in \mathcal{T}^{\leafset}$, the tree $\bar{T}$ is a {\bf $g$-grafting} of $T$ if $\leafset(\bar{T}) \backslash \leafset(T) = \{g \}$%, the edge ending in leaf $g$ in $\bar{T}$ has length $x \leq 0$ \megan{$>$ or $\leq$?}
, and $\Psi_{\bar{g}}(\bar{T}) = T$.
\end{defi}

In other words, a grafting of $T$ consists of a tree identical to $T$, but with one additional leaf $g$ and its leaf edge $e_g$. In considering the possibilities for such a grafting, there are two independent choices: the non-negative length of $e_g$, and a point on $T$ at which to graft the non-leaf end. The next lemma shows the consequences of this, and a bit more.

 \begin{lemma}
 \label{l:grafting}
 For tree $T \in \cT^{\leafset}$ and leaf $g \notin \leafset$, the space of $g$-graftings of $T$, denoted $\Psi_{\bar{g}}^{-1}(T)$, is the direct product of $\mathbb{R}_{\geq 0}$ and a piecewise-linear connected curve which is graph-isomorphic to $T$ and which intersects a strict subset of orthants each in a 1-dimensional linear curve.
 %the direct product of a piecewise-linear connected curve which is graph-isomorphic to $T$ and intersecting a strict subset of orthants each in a 1-dimensional linear curve, and $\mathbb{R}^+$.
\end{lemma}
\begin{proof}
Consider any tree $T \in \cT^{\leafset}$, leaf $g \notin \leafset$ and length $x \geq 0$.  Recall that $\cE(T)$ is the set of edges of tree $T \in \cT^{\leafset}$, with each edge $e \in \cE(T)$ having split $P_e$ and length $w_e$. 

%Let $ w_1, P_1,w_2,P_2\dots,w_k,P_k,w_{k+1},P_{k+1},\dots,w_{2k-3},P_{2k-3}$ be the respective lengths and partitions for each edge, including internal and leaf edges, in a $k$-tree $T$.

We can attach a new edge $e_g$ of length $w_g$ ending in leaf $g$ to any point, including an endpoint, on any edge of $T$ to get a $g$-grafting of $T$.  Thus the set of $g$-graftings of $T$, $\Psi_{\bar{g}}^{-1}(T)$, is not empty.  For any $\bar{T} \in \Psi_{\bar{g}}^{-1}(T)$, its additive metric $A_{\bar{T}}$ restricted to the leaves $\leafset$ is just the additive metric of $T$, $A_T$. 
%
%If $\bar{T}$ is a $(k+1,x)$-grafting, then $\bar{T}$ contains an isometric copy of $T$ (as a metric space) on the same label set $\{1,2,\dots,k\}$, and the $(k+1)$-edge intersects the copy of $T$ at some point on $T$.
%
It follows $\bar{T}$ can be completely characterized by two independent choices: the choice of point on $T$ for grafting, the space of which is graph-isomorphic to $T$, and a choice of length for the grafted leaf edge, which can be any non-negative real number.

%If $\bar{T}$ is a full $(k+1,x)$-grafting, then the $(k+1)$-edge intersects the tree in the interior of some edge. 
Let $e \in \cE(T)$ be the edge to which $e_g$, which has split $\bar{P}_g = (g)(\leafset)$, will be grafted to form $\bar{T}$.  If we are grafting $g$ to a vertex of $T$, then choose $e$ to be one of the edges adjacent to this vertex.  For each edge $f \in \cE(T) \backslash e$, the two partitions of the leaves in the corresponding split $P_f$ induce two subtrees of $T$, and edge $e$ is completely contained in one of these subtrees.  Add leaf $g$ to the partition of $P_f$ corresponding to this subtree to get $\bar{P_f}$, the corresponding split in $\bar{T}$.  The split $P_e$ becomes the splits $\bar{P_e}^L = (P_e)(P_e^c \cup g)$ and $\bar{P_e}^R = (P_e \cup g)(P_e^c)$ in $\bar{T}$.  If $e_g$ was grafted to an endpoint of $e$, then one of $\bar{P_e}^L, \bar{P_e}^R$ will have zero weight, but we will still include it here as a split for consistency. Thus $\bar{T}$ has precisely the splits 
$\{\bar{P_f}: f \in \cE(T) \backslash e\} \cup \bar{P}_g \cup \bar{P_e}^L \cup \bar{P_e}^R$.  
%$P_1',P_2',\dots, \hat{P_i}',\dots , P_{2k-3}', P^{(k+1)}_{2k-2}$, where $P_j'$ is $P_j$ with the label $k+1$ in the same partition as the descendants of the endpoints of the $P_i$ edge, and $P^{(k+1)}_{2k-2}$ is the leaf split $(\{1,2,\dots,k\},\{k+1\})$. 
%Then instead of $P_i$, we have two new edges, as $P_i$ is split into two, one $P_i^L = (P_i \cup k+1) (P_i^c)$, and $P_i^R = (P_i)(P_i^c \cup k+1)$. This fully characterizes the splits of $\bar{T}$. 

For each edge $f \in \cE \backslash e$, the weight of split $\bar{P_f}$ in $\bar{T}$ is the same as the weight of split $P_f$ in $T$, since the edge corresponding to $\bar{P_f}$ projects to the edge corresponding to $P_f$ without distortion.  Thus, we will represent the weight of edge $f$ in $\bar{T}$ by $w_f$ as well.  Split $\bar{P_g}$ has weight $w_g$, and let splits $\bar{P_e}^L$ and $\bar{P_e}^R$ have weights $w_e^L$ and $w_e^R$, respectively. Then the space of all $\bar{T}$ formed by grafting leaf $g$ to edge $e$ is a two-parameter family satisfying $w_e = w_e^L + w_e^R$, and $w_g, w_e^L,w_e^R \geq 0$.  Note that $w_g$ is a free parameter, and $w_e = w_e^L + w_e^R$ is the equation of a line.  Thus this solution space in this orthant is the direct product of $\RR_{\geq 0}$ with the line that intersects the orthant boundaries at $w_e^L = 0, w_e^R = w_e$ and at $w_e^L = w_e, w_e^R = 0$. 

%Then the space of all such $\bar{T}$ is a two-parameter family satisfying $w_f' = w_j$ for $j =1,2,\dots,\hat{i},\dots,2k-3$, $w_i = w_i^L + w_i^R$, and $w_{2k-2}, w_i^L,w_i^R \geq 0$, where the $w_e$'s represent the lengths of their respective partitions as before, as $P_j'$ projects to $P_j$ without distortion and $P_i^L \cup P_i^R$ both project to $P_i$ with additive length, and $P_{2k-2}$ is a free parameter. So if all $w_j$ and $w_i$ are fixed, then we have a two dimensional linear space of pre-images of $\Psi_{\bar{k+1}}$ in this particular orthant, intersecting the $w_i^L = 0$ and $w_i^R = 0$ hyperfaces in $w_i^R = w_i$ and $w_i^L = w_i$, respectively.

It remains to show that the lines given by $w_e^L + w_e^R = w_e$ in each orthant are connected and graph isomorphic to tree $T$.  Let $e$ and $e'$ be two adjacent edges in $T$, separated by vertex $v$.  Edges $e$ and $e'$ are compatible because they exist in the same tree, and thus the intersection of one partition from each split is empty.  Without loss of generality (by temporarily renaming the partitions if necessary), assume that $P_e \cap P_{e'} = \emptyset$.  Then the case $w_e^L = w_e, w_e^R = 0$ corresponds to a tree with splits $\bar{P_e}^L = (P_e)(P_e^c \cup g)$, with weight $w_e$, and $\bar{P_{e'}} = (P_{e'})(P_e \cup g)$, with weight $w_{e'}$, as well as splits $\bar{P_f}$, with weight $w_f$, for all $f \in \cE(T) \backslash \{e,e'\}$, and $\bar{P}_g$, with weight $e_g$.  The case $w_{e'}^L = w_{e'}, w_{e'}^R = 0$ corresponds to a tree with splits $\bar{P_{e'}}^L = (P_{e'})(P_{e'}^c \cup g)$, with weight $w_{e'}$, and $\bar{P_e} = (P_e)(P_e \cup g)$, with weight $w_e$, as well as splits $\bar{P_f}$, with weight $w_f$, for all $f \in \cE(T) \backslash \{e,e'\}$, and $\bar{P}_g$, with weight $e_g$.  But these are identical split and weight sets, and thus the two line endpoints coincide.  Since the two of these line segments meet if and only if they correspond to attaching leaf $g$ to adjacent edges in $e$, we get that the piecewise-linear connected curve is graph-isomorhpic to $T$.
\end{proof}

\begin{exmp}
\label{ex:running_start}
Suppose we have a tree $T$ with labels $\{1,2,3,5\}$ as depicted in Figure~\ref{fig:trees_3}, with leaf edges having length $\{0.15,0.3,0.2,0.25\}$ respectively, and interior edge length $0.2$. The corresponding additive distance matrix (indexed respectively) is given by
$$ A_{T} = \left(
\begin{array}{l l l l }
0 & .65 & .35 &  .6 \\
.65 & 0 & .7  &  .55\\
.35 & .7 & 0 &  .65 \\
.6 & .55 & .65 &  0 
\end{array}\right)$$ Then the preimage of $\Psi_{\bar{4}}$ is the product of the subspace of $\mathcal{T}^5$ depicted on the right in Figure~\ref{fig:trees_3} (with leaf edge length for $1,2,3,5$ determined uniquely by the point on $\Psi_{\bar{4}}(T)$ below) and the copy of $\mathbb{R}_{\geq 0}$ (not shown) representing the ``4"-leaf edge length. If we fix the length $y$ of the 4 leaf, the $(4,y)$-grafting of $T$ is the subspace shown by a thick line, together with unique local leaf coordinates $$(w_1,w_2,w_3,w_4,w_5) = (0.15-x_{(14)}, 0.3-x_{(24)}, 0.2 - x_{(34)}, y, 0.25- x_{(45)})$$
where $x_{(14)}, x_{(24)}, x_{(34)}, x_{(45)}$ are the weights of splits $(14),(24),(34),(45)$, respectively, if that split exists in the tree, and 0 otherwise. 

While it may appear that the four line segments corresponding to grafting $g$ to a leaf edge end mid-orthant, this is only because the figure omits the dimensions of those orthants corresponding to the leaf edges.  The line segments ends on boundaries where the respective leaf edge lengths are 0.

\begin{figure}
    \centering
    \includegraphics[scale=0.5]{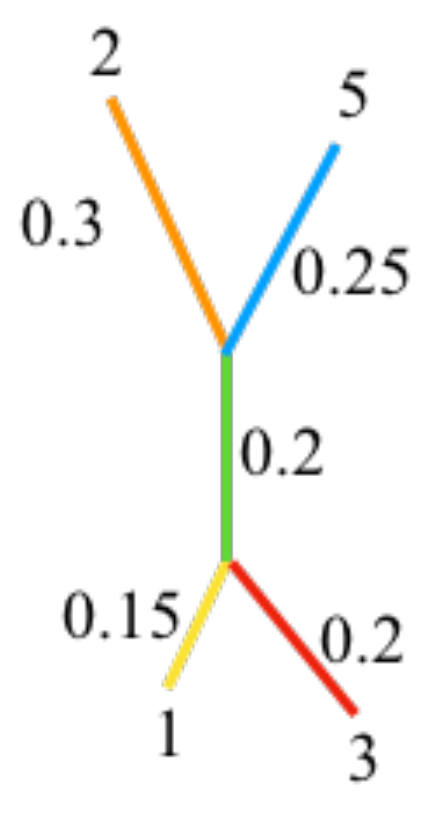} \quad\quad\quad
    \includegraphics[scale=0.5]{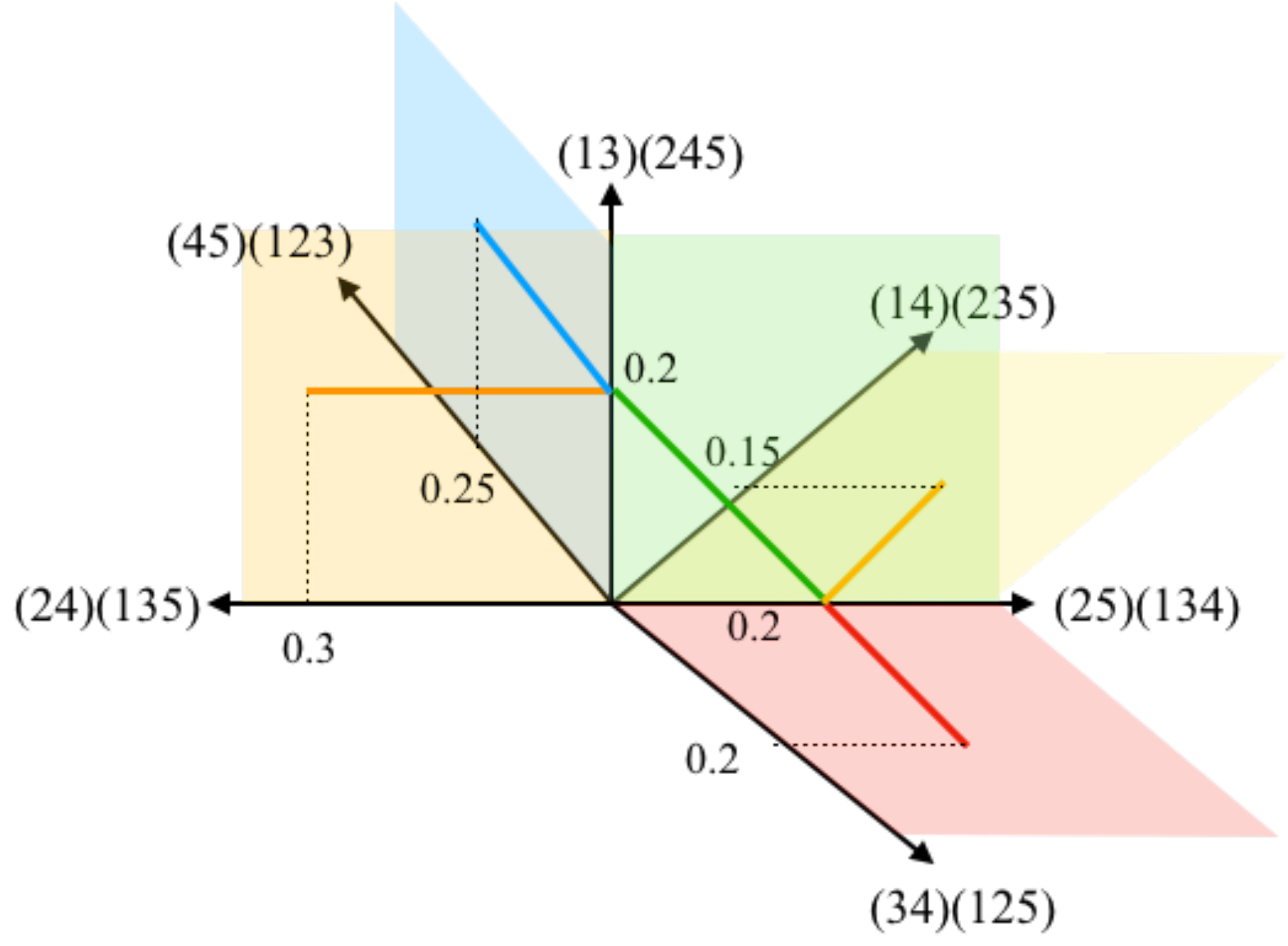}
    \caption{Left, a tree $T$ with 4 leaves, $\{1,2,3,5\}$. Right, the orthants of $\cT^5$ containing the preimage $\Psi^{-1}_{\bar{4}}(T)$, with the subspace corresponding to the preimage shown with the thick solid lines. Note that the dimensions corresponding to the 4 leaf edges lengths were not included for clarity.}
    \label{fig:trees_3}
\end{figure}

\end{exmp}

\subsection{Extension by Multiple Leaves}

As defined in \cite{UIUC}, the {\bf connection cluster} $C_{S(T),n,\ell}$ of a tree topology $S(T)$ on leaf set $[n] =\{1,2,\dots,n\}$ is the set of binary tree topologies with $n+\ell$ leaves obtained from adding $\ell$ leaves to arbitrary edges of $T$. We will generalize the definition of a connection cluster to allow the leafset $\leafset$ of $T$ to be any subset of $[N] = \{1, 2, ..., N\}$, and use the notation $C^N_T$, where $T \in \cT^{\leafset}$ and $\leafset \subset N$.  Throughout this section, we will still assume that $|\leafset| = n$, and $N = n + \ell$.  %Since the leaf set and split set are given in the presentation of $T$, we will often suppress the $n$ and $\ell$, and write $C^N_T$ to denote the extension of $S(T)$ to $\mathcal{T}^N$, for $N = n+\ell$. 
The {\bf connection space} $S_{S(T),n,\ell}$ in the notation of \cite{UIUC}, or $S^N_T$ in our notation, is the union of the closed orthants in $\cT^N$ that represent the elements of $C_N^T$, i.e. a non-negative real orthant for every unweighted tree in $C_N^T$ under the normal identification of faces. The {\bf connection graph} $G_{S(T),n,\ell}$, or with a change of notation, $G^N_T$, is the intersection of $S^N_{T}$ with the link $L_{N}^1$, in which maximal cliques give elements of $C^N_T$.  In \cite{UIUC} and Lemma 3.7 below, it is shown that the edges of a connection graph are determined by normal pairwise compatibility of splits in $\mathcal{T}^N$, which allows for quick computation of $C_N^T$.

The connection space $S_T^N$ can also be seen as the preimage in $\mathcal{T}^{N}$ under $\Psi_{\leafset}$ of the entire orthant represented by $S(T)$, namely $\Psi^{-1}_{\leafset}(\cO(T)).$ %, under contraction of the $\ell$ additional leaves specified by tree dimensionality reduction, and  
Similarly, the connection graph $G_T^N$ is the corresponding preimage of the complete $n$-graph on $S(T)$.  We are then interested in the subspace of $S_T^N$, restricted by the edge lengths of $T$, which projects under tree dimensionality reduction to $T$. This subspace will be a $2\ell$-dimensional linear submanifold supported in $S_T^N$. In other words, once the combinatorics of the extended trees are calculated through the connection cluster, 
%once the combinatorics of the dimension increase are calculated, 
we can use a set of $(2n-3)$ linear equations parametrized by the edge lengths in $T$ %the $n$-tree 
to constrain sums of fixed edges in $\cT^N$ %$(n+\ell)$-space
, and give the complete preimage $\Psi_{\leafset}^{-1}(T)$.

\subsection{Calculating the Metric Extension Space}
In this section we will construct, for phylogenetic tree $T \in \mathcal{T}^n$, the subset $E_T^N \subset S_N^T \subset\mathcal{T}^N$ which results from gluing $\ell$ leaves of arbitrary length to the metric tree $T$. The computation of the extenstion space $E_T^N$ has two steps: 

The first step is the computation of $S_T^N$, via the method in \cite{UIUC} for constructing $G_T^N$ and $C_T^N$. We will see that this is the preimage under $\Psi_{\leafset}$ of the orthant containing $T$.

 The second step introduces the constraint that under the action of $\Psi_{\leafset}$ on $S_T^N$, the process of deleting and concatenating edge lengths as described in Definition \ref{TDR} yields $T$ precisely. To find the trees which satisfy this constraint, we solve a system of linear equations separately for each orthant in $S_T^N$.

\subsubsection{Combinatorial Step}\label{extcombstep}
As in the previous section, we let $\{P_e\}_{e \in \cE(T)}$ be the splits of $T$ (including the leaf edges), with corresponding lengths $\{w_e\}_{e \in \cE(T)}$. We will first state the algorithm for computing the connection cluster $C^T_N$ and give an example, before proving correctness.\\

\hspace{-1.5em}{\bf Computation of 
Connection Cluster}
\begin{enumerate}
\item For each $P_e$, construct the set $\mathbf{Q}_e$ of splits projecting to $P_e$ by adding the $\ell$ labels $N \backslash \leafset$ to $P_e$ or $P_e^c$ in all possible $2^\ell$ ways. 

\item Take the union $\mathbf{Q} = \cup_{e \in \cE(T)} \mathbf{Q}_e$ to get the vertices of the connection graph $G_T^N$.  Add an edge between each pair of vertices if and only if the two splits are compatible, which can be checked by the condition given in Definition~\ref{compatible}.

\item Find all maximal ($n+\ell-3$) cliques in the subgraph of thick partitions, which is found by removing the leaf splits, $\{Q \in \mathbf{Q}:|Q| = 1 \}$. Extend each maximal clique to include the leaf partitions, which are compatible with all other partitions, and return the corresponding set of cliques $C^N_T$.
\end{enumerate}

\begin{exmp} Returning to the tree in Example 3.4, we find $C^5_T$ using the above algorithm. The set of splits $S(T) =\{25|13,1|235,2|135,3|125,5|123\}$
, so in Step 1, we find the set
$$\mathbf{Q} =\{13|245,25|134,14|235,24|125,34|125,45|123,1|2345,2|1345,3|1245,4|1235,5|1234\}$$
In the second step, we form the graph $G^5_T$, which is shown in Figure~\ref{fig:connection_graph}. 

\begin{figure}
    \centering
    \includegraphics[scale=0.5]{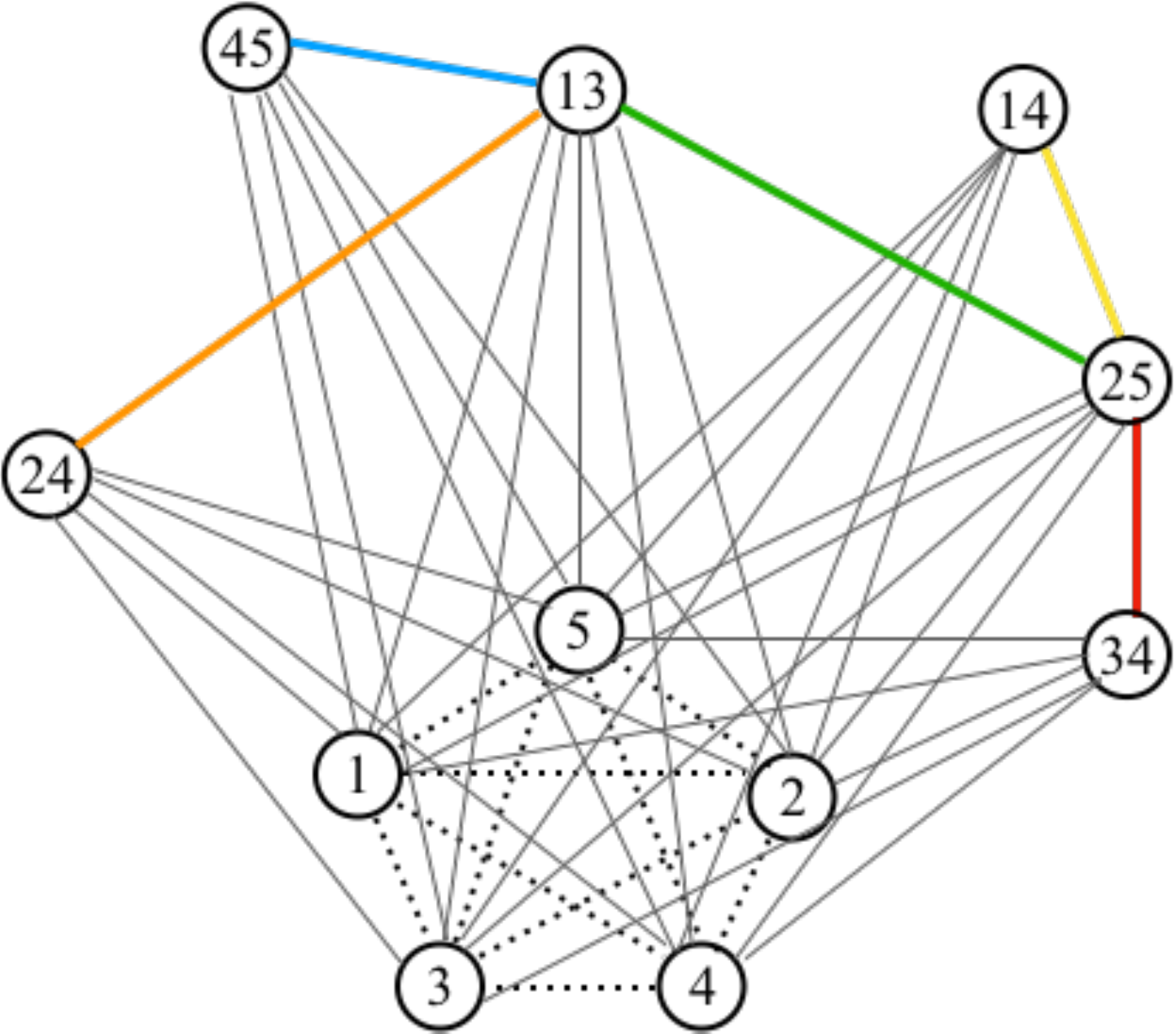}
    \caption{The connection graph $G^5_T$ for tree $T$ from Example~\ref{ex:running_start}.  The vertices corresponding to elements of $\mathbf{Q}$ are labeled by the smaller of the two pieces of the partition. The leaf partitions have automatic compatibility - these edges are shown dotted, while compatible thick partitions have colored edges. }
    \label{fig:connection_graph}
\end{figure}

%\begin{center}
%\begin{tikzpicture}
%  [scale=.8,auto=left]
%  
%  \node (45) at (-3,5) {(45)};
%  \node (n1) at (1,1)  {};
%  \node (13) at (1,5)  {(13)};
%  \node (25) at (6,1)  {(25)};
%  \node (24) at (-4,1) {(24)};
%
%  \node (14) at (5.5,4) {(14)};
%  \node (34) at (6,-1.5) {(34)};
%  \node (l1) at (-.25,-2.75) {(1)};
%  \node (l2) at (2.25,-2.75) {(2)};
%  \node (l3) at (0,-4) {(3)};
%  \node (l4) at (2,-4) {(4)};
%  \node (l5) at (1,-1.5) {(5)};
%  
%  \draw[thick,blue!50!white] (45) to (13);
%  \draw[thick,green!70!white] (13) to (25);
%  \draw[thick,yellow] (14) to (25);
%  \draw[thick,orange!60!white] (24) to (13);
%  \draw[thick,red] (34) to (25);
%  
%  \foreach \from/\to in {l1/45,l1/13,l1/25,l1/24,l1/14,l1/34,l2/45,l2/13,l2/25,l2/24,
%  l2/14,l2/34,l3/45,l3/13,l3/25,l3/24,l3/14,l3/34,l1/l2,l1/l3,l1/l4,l1/l5,l2/l4,l2/l5,l2/l3,l3/l5,l3/l4,l4/l5,
%  l4/45,l4/13,l4/25,l4/24,l4/14,l4/34,l5/45,l5/13,l5/25,l5/24,
%  l5/14,l5/34}
%    {  \draw[dotted] (\from) to (\to); }
%\end{tikzpicture}
%\end{center}

In Step 3, we find maximal $(4+1-3)$-cliques in the thick subgraph. The $2$-cliques are edges, and for each edge, we can include all of the leaf edges to that set to obtain a unique topology of $\mathcal{T}^5$. All such topologies form the connection cluster $C^5_T$.  The orthants corresponding to these topologies are precisely those pictured in Example 3.4, and form $S^5_T$, the connection space, which is shown again in Figure~\ref{fig:connection_space} without the leaf dimensions.
%The resulting orthants, $C^5_T$, are precisely those pictured in Example 3.4, the {\em orthant support} of the extension space, shown again below without the leaf dimensions. 

\begin{figure}
    \centering
    \includegraphics[scale=0.5]{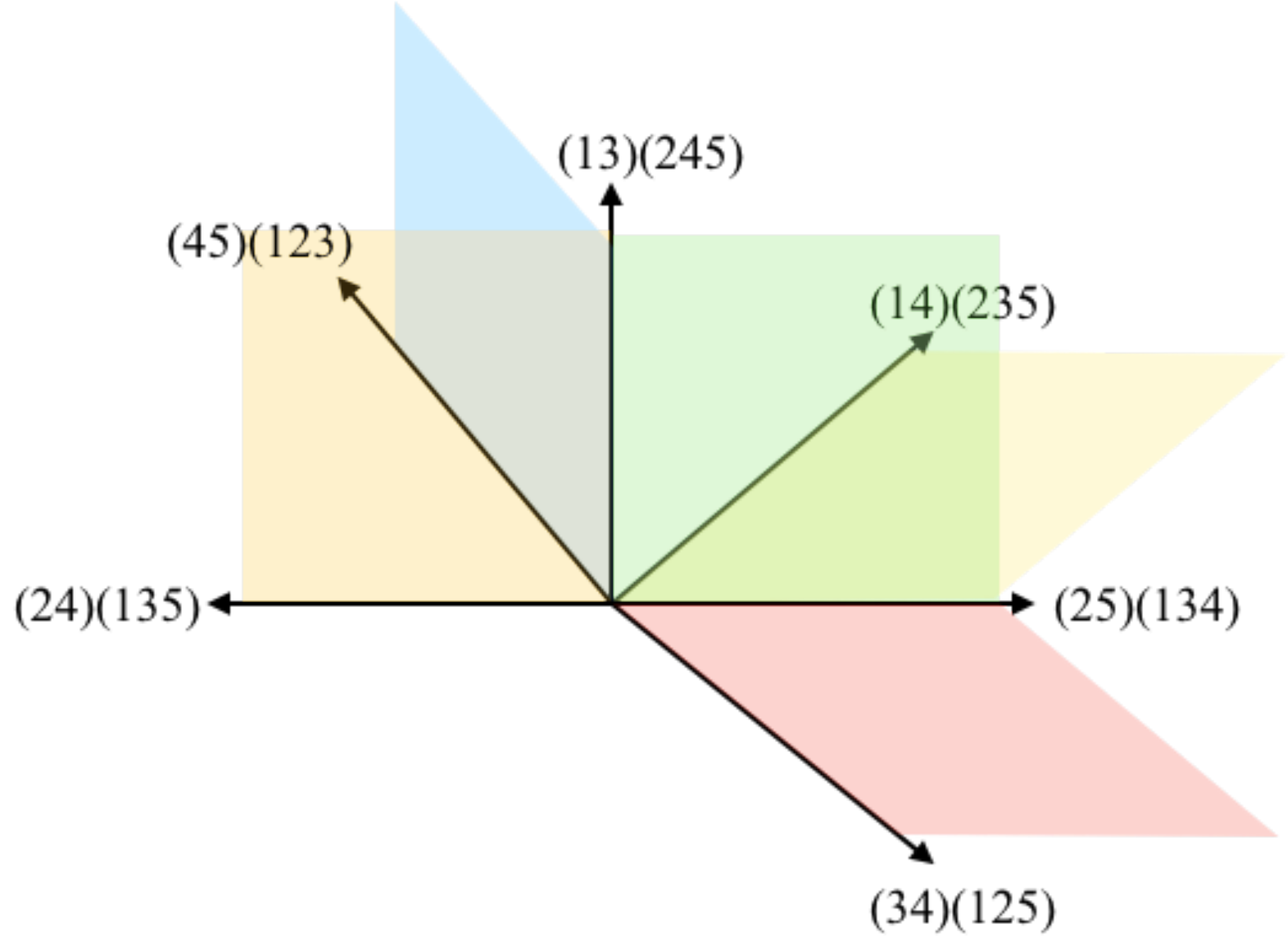}
    \caption{The connection space $S^5_T$ for tree $T$ from Example~\ref{ex:running_start}. }
    \label{fig:connection_space}
\end{figure}

%\begin{center}
%\begin{tikzpicture}
%  [scale=.85,auto=left]
%  \draw[fill=yellow!10,fill opacity=0.5, draw=none] (1,1) -- (6,1) -- (10.5,4) -- (5.5,4)--cycle;
%  \draw[fill=green!10, fill opacity=0.5,draw=none] (1,1) -- (6,1) -- (6,5) -- (1,5)--cycle;
%  \draw[fill=red!10,fill opacity=0.5,draw=none] (1,1) -- (6,1) -- (11,-1.5) -- (6,-1.5)--cycle;
%  \draw[fill=blue!10,fill opacity=0.5,draw=none] (1,1) -- (-3,5) -- (-3,9) -- (1,5)--cycle;
%  \draw[fill=orange!10,fill opacity=0.5,draw=none] (1,1) -- (-4,1) -- (-4,5) -- (1,5)--cycle;
%  \node (n0) at (-3,5) {(45)(123)};
%  \node (n1) at (1,1)  {};
%  \node (n7) at (1,5)  {(13)(245)};
%  \node (n8) at (6,1)  {(25)(134)};
%  \node (n10) at (-4,1) {(24)(135)};
%  \node (14) at (5.5,4) {(14)(235)};
%  \node (34) at (6,-1.5) {(34)(125)};
%
%  
%  \foreach \from/\to in {n1/n0,n1/n7,n1/n8,n1/n10,n1/14,n1/34}
%    {  \draw (\from) to (\to); }
%\end{tikzpicture}
%\end{center}
\end{exmp}

The proposition below shows that the set of cliques returned in the final step of the algorithm is indeed the connection cluster $C_T^N$, justifying the notation.

\begin{prop}\label{CT} For $T\in \mathcal{T}^{\leafset}$ with $\leafset \subset [N]$, the above algorithm returns the cliques $C_T^N$, which correspond to the orthant support of $\Psi_{\leafset}^{-1}(T)\subset \mathcal{T}^{N}$. 
\end{prop}

First we show a preliminary result allowing us to reduce to conditions on the vertices of the extension graph.

\begin{lemma}
\label{l:tdr_on_splits}
For tree $T \in T^{\leafset}$ with $\leafset \subset [N]$, an orthant $\cO \subset \mathcal{T}^{N}$ contains an element of $\Psi_{\leafset}^{-1}(T)$ if and only if $\Psi_{\leafset}(S(\cO))= S(T)$. That is, $\cO$ contains a tree in the extension space of $T$ if and only if removing the labels $N \backslash \leafset$ from the splits $S(\cO)$ yields precisely the split set of $T$ (with multiplicity).
\end{lemma}
\begin{proof} We proceed by induction on $\ell = |N \backslash \leafset|$.

If $\ell = 1$ and $\bar{T}$ is an extension of $T\in \mathcal{T}^{\leafset}$ by grafting leaf $g$ to edge $e \in \cE(T)$, then from the proof of Lemma~\ref{l:grafting}, $\bar{T}$ has split set $S(\bar{T}) = \{\bar{P_f}: f \in \cE(T) \backslash e\} \cup \bar{P}_g \cup \bar{P_e}^L \cup \bar{P_e}^R$.  Recall that removing edge $f$ from $T$ induces two subtrees, the vertices of which become the two parts of splits $P_f$, and that  $\bar{P_f}$ was constructed from $P_f$ by adding leaf $g$ to the partition corresponding to the subtree to which $g$ was grafted.  Thus $\bar{P_f}$ projects to $P_f$ by construction for all $f$.  Similarly, $\bar{P_e}^L$ and $\bar{P_e}^R$ were constructed such that they project unto $P_e$.  Finally $\bar{P_g}$ projects onto a split with one partition empty, which we delete.

%As in Lemma 3.3, if $\ell = 1$ and $\bar{T}$ is an extension of $T\in \mathcal{T}^{n}$ (by grafting), then $\bar{T}$ has split set $S(\bar{T}) = P_1',P_2',\dots,P_{i-1}',P_i^L,P_i^R,P_{i+1}',\dots, P_{2n-3}', K$, $K = (n+1)(1,2,\dots,n)$, where the $n+1$ label is placed in one half of the split or the other, without perturbing the labels $1,2,\dots,n$, according to the rule: $P_j$ induces subtrees with leaf sets $P_j$ and $P_j^c$ by removing the edge $e_j$, one of which contains the edge given by split $P_i$. Without loss of generality, if this is the $P_j^c$ subtree, then add the label $n+1$ to the $P_j^c$ half of the partition. Thus each edge $P_j'$ projects to $P_j$, and $P_i^L$ and $P_i^R$, which are obtained by adding the $n+1$ label alternately to each half of $P_i$, both project to $P_i$, by construction. 

Conversely, if a set $S$ of pairwise-compatible splits on $[N]$ projects to $S(T)$ under deletion of some leaf $g = N \backslash \leafset$, then we claim there exists a unique split $P/P^c \in S(T)$ which has two preimages.  Suppose not.  That is, suppose for $P/P^c$ and $Q/Q^c$ splits in $T$, the collective split preimages are $(P\cup g)(P^c)$, $(P)(P^c \cup g)$, $(Q \cup g)(Q^c)$, and $(Q)(Q^c \cup g)$. Then compatibility of $P$ and $Q$ in $T$ guarantees that precisely one of $Q\cap P, Q^c\cap P, Q \cap P^c, Q^c \cap P^c$ is empty, say without loss of generality $Q\cap P$. Then $(Q \cup g)(Q^c)$ and $(P \cup g)(P^c)$ are not compatible, because none of the four intersections of their partitions are empty.  Thus $S$ contains only one of them. So for any pair of splits in $T$, there are at most 3 preimage splits in $S$, and unique splits have distinct preimages, so we conclude that there is a unique split in $T$ with both preimages, i.e. the set $S$ must look precisely as above, $\{\bar{P_f}: f \in \cE(T) \backslash e\} \cup \bar{P}_g \cup \bar{P_e}^L \cup \bar{P_e}^R$, and from this we can construct $\bar{T}\in \Psi_{\leafset}^{-1}(T)$ uniquely by grafting the $g$-leaf edge to the middle of edge $e$.

%Conversely, if a set $S$ of $2n-3$ pairwise-compatible splits projects to $S(T)$ under deletion of the $n+1$-leaf, then there exists a unique split $P/P^c$ which has two preimages: for $P/P^c$ and $Q/Q^c$ splits in $T$, the collective split preimages are $P\cup (n+1)$, $P^c \cup (n+1)$, $Q \cup (n+1)$, and $Q^c \cup (n+1)$. Then compatibility of $P$ and $Q$ in $T$ guarantees that precisely one of $Q\cap P, Q^c\cap P, Q \cap P^c, Q^c \cap P^c$ is empty, say without loss of generality $Q\cap P$. Then $Q\cup (n+1)$ and $P\cup (n+1)$ are not compatible, implying that $S$ contains only one of them. So for any pair of splits in $T$, there are at most 3 preimage splits in $S$, and unique splits have distinct preimages, so we conclude that there is a unique split in $T$ with both preimages, i.e. the set $S$ must look precisely as above, $P_1',\dots,P_i^L,P_i^R,\dots, P_{2n-3}',K$, and from this we can construct $\bar{T}\in \Psi_n^{-1}(T)$ uniquely by grafting the $(n+1)$-leaf edge to the middle of edge $e_i$.

So we have the result for the $\ell = 1$ case. 

Then assume for induction that there exists $\bar{T}\in \cO \subset \mathcal{T}^{n + \ell}$ such that $\Psi_{\leafset}(\bar{T}) = T$, if and only if $\Psi_{\leafset}(S(\cO))=S(T)$. Then let $\cO'$ be an orthant in $\mathcal{T}^{n+\ell+1}$.  So then $\Psi_{n+\ell}(\cO')$ is an orthant in $\mathcal{T}^{n+\ell}$, and applying the inductive hypothesis, there exists $\bar{T}'\in \Psi_{n+\ell}(\cO')$ with $\Psi_{\leafset}(\bar{T}') = T$ if and only if $\Psi_{\leafset}(S(\Psi_{n+\ell}(\cO'))) = S(T)$. Since $S(\Psi_{n+\ell}(\cO')) = \Psi_{n+\ell}(S(\cO'))$ from the one-step case, and $\Psi_{\leafset}(\Psi_{n+\ell}(S(\cO'))) = \Psi_{\leafset}(S(\cO'))$, giving us the forward direction. %we find $\bar{T}\in S$ such that $\Psi_{\leafset}(\bar{T}) = T$.
For the reverse direction, we know that $\bar{T}' \in \Psi_{n+\ell}(\cO'),$
%But we know that $\bar{T}' \in \Psi_{n+\ell}(S),$ 
which means that there is some tree $\bar{T}\in \cO$ such that $\Psi_{n+\ell}(\bar{T}) = \bar{T}'$ by the base case. For this tree, then, $\Psi_{\leafset}(\bar{T}) = \Psi_{\leafset}\Psi_{n+\ell} \bar{T} = \Psi_{\leafset} \bar{T}' = T$, and the proof is complete.

\end{proof}

\begin{proof} (of Proposition) Suppose we have a maximal clique in $G_N^T$. Then this clique represents a set of pairwise compatible splits. Since $L_n^1$ is a flag complex, these splits represents an orthant $\cO$ in $\mathcal{T}^{N}$, of dimension corresponding to the size of the clique. By Lemma~\ref{l:tdr_on_splits}, these splits projects to the splits of $T$, so the orthant $\cO$ contains elements of the extension space.

Conversely, suppose a tree $\bar{T}$ is in the extension space. Then by Lemma~\ref{l:tdr_on_splits}, the splits of $\bar{T}$ are among the vertex set of $G_N^T$, and since $\bar{T}$ is a tree in $\mathcal{T}^{N}$, its splits are compatible. Since this is the condition for connectivity in $G_N^T$ as well as $L_n^1$, $\bar{T}$ maps to a clique in $G_N^T$. 
\end{proof}

\begin{prop}\label{connection runtime}This algorithm is $O(2^{3\ell}n^3)$.
\end{prop}

\begin{proof}
In the first step, we do a simple enumeration, with run time $(2n-3)2^\ell$. The second step of removing duplicates and initializing the graph is then $O(2^{2\ell}n^2)$, and to check compatibility is $O(2n-3+\ell)$ in each pair, so has $O(2^{2\ell}n^3)$. By \cite{TIAS}, the run time of maximal clique enumeration is $O(|E|*|V|)$, and from \cite{UIUC} we have that the vertex set has size $2^{\ell}(2n-2)-\ell-n-1$, and the edge set size being at most the square of this, we have a $O(2^{3\ell}n^3)$ run time for clique enumeration. This dominates the other steps, which gives the result.
\end{proof}

Note that while this is fairly quick in $n$, it may be the case that we have small fragments of large trees, implying a very dominant $\ell$ term. In this case, the algorithm is essentially reconstructing a large portion of $\mathcal{T}^{n+\ell}$, and so there is not much improvement which can be made, since the solution space itself is large. In the next section we will address a method for handling small tree fragments among a set of tree fragments.

\subsubsection{Metric Step}
%Given a clique $K \simeq K_{n+\ell}$ in $G_{S(T),n,\ell}$, $K$ gives 
Consider an orthant $\cO \subset S^N_T \subset \mathcal{T}^N$, and index its corresponding splits by $Q_1,Q_2,\dots,Q_{2N-3}$ (for example, in lexicographical order).
%, where the split set of $K$ is extended to include leaf partitions, and with slight abuse of notation in the indexing. 
By construction, $\Psi_{\leafset}(Q_j) = P_i$ for some $i\in\{1,\dots, 2n-3\}$. We represent this assignment with a $(2n-3)\times (2N-3)$ {\bf projection matrix} $M^{\cO}_T = (m_{ij})$, where 
$$m_{ij} = \left\lbrace\begin{array}{l l}1 & \mbox{ if } \Psi_{\leafset}(Q_j) = P_i\\
0 & \mbox{otherwise}\end{array}\right.$$
 Since this is a well-defined map from $\{Q_j\}$ to $S(T)= \{P_i\}$, columns each have a unique non-zero entry. 
We then set up the real system of equations:
\begin{equation}\label{metconstraints} 
\begin{array}{l}
M^{\cO}_T \cdot \mathbf{x}^{\cO} = \mathbf{w}\\
\mathbf{x}^{\cO}\geq 0 \end{array}
\end{equation}
for $\mathbf{x}^{\cO}$ the vector of non-negative edge weights in $\cO$ ($x_j$ the weight of split $Q_j$), and $\mathbf{w}$ the vector of edge weights in $T$.

To see what this system is producing, notice that this specifies, for each split $P_i$ in $T$ with weight $w_i$, the equation
$$ x_{j_1}+x_{j_2}+\dots+x_{j_{a_i}} = w_i$$
for $Q_{j_1},\dots, Q_{j_{a_i}}\in S(\cO)$ projecting to $P_i$, so that under tree dimensionality reduction $\Psi_{\leafset}$, the (non-negative) lengths of the edges $e'_{j_1}, e'_{j_2}, ..., e'_{j_{a_i}}$ of a tree in $\cO$ concatenated to produce edge $e_i\in T$ sum precisely to $w_i$. So solving this system of equations would find vectors of possible edge lengths in tree topologies which project to $T$.

\begin{defi} Given an orthant $\cO \in S_T^N \in \mathcal{T}^{n+\ell}$, which, alternatively, has splits corresponding to a clique in $G_T^N$ and a topology in $C_T^N$, we call the set of $\mathbf{x}^{\cO}$ satisfying (\ref{metconstraints}) the {\bf extension space of $T$ in $\cO$}, denoted $E^{\cO}_{T,n,\ell}$ or $E^{\cO}_T$. The {\bf extension space of $T$ in $\mathcal{T}^N$} is defined to be the union of extension spaces over all orthants in the connection space:
$$ E_{T,n,\ell} := \bigcup_{\cO \in S^N_T} E^{\cO}_{T,n,\ell}.$$
\end{defi}

Note that the image of $\mathbf{Q} = \{Q_1,\dots, Q_{2N-3}\}$ under tree dimensionality reduction to $\leafset(T)$ gives a partition of the set into precisely $2n-3$ components, because $\Psi_{\leafset}(\mathbf{Q})$ is well-defined and surjective on $P_i$'s. Because it is a partition and $w_i > 0$, we are guaranteed a solution of dimension $\sum_{j}m_{ij}-1$ to the equation above, and a total solution space of dimension 
$$\sum^{2n-3}_{i=1} ((\sum^{2N-3}_{j=1} m_{ij}) - 1) = \sum^{2N-3}_{j=1} \sum^{2n-3}_{i=1} m_{ij} - (2n-3) = (2N-3)-(2n-3) = 2\ell.$$
This generalizes the single leaf extension case in that, after the equations are solved for all orthants, the result is the direct product of a piecewise-linear connected $\ell$-manifold (intersecting a strict subset of orthants each in an $\ell$-dimensional linear subspace), with $(\mathbb{R}_{\geq 0})^\ell$. Connectivity follows from the consideration that if two orthants share a $k$-dimensional face, then that face is represented as a $k$-clique in the connection graph, and the metric extension space meets the face in a set of equations of precisely the same sort on each side. %, for $Q_{j_1},Q_{j_2},\dots, Q_{j_k}$ for $k < 2N-3$.  

\begin{prop}
For leafset $\leafset \subset [N]$, let $T \in \cT^\leafset$ be a binary tree.  The extension space of $T$, $E_T^N$, is connected.  Furthermore, for adjacent orthants $\cO_1, \cO_2 \subset S_T^N$, $E^{\cO_1 \cap \cO_2}_T = E^{\cO_1}_T \cap \cO_2 = \cO_1 \cap E^{\cO_2}_T$. 
\end{prop}
\begin{proof}
For each orthant $\cO \subset S_T^N$, the extension space $E_T^{\cO}$ is connected, since it is the solution of a linear system of equations, restricted to the non-negative orthant.  Any two adjacent orthants $\cO_1, \cO_2 \subset S_T^N$ share some $k$-dimensional boundary orthant, which corresponds to a $k$-clique in the connection graph. 
Suppose the $k$ splits in the clique are $Q_{j_1},Q_{j_2},\dots, Q_{j_k}$.  Then any solutions $\mathbf{x}^{\cO_1}$, $\mathbf{x}^{\cO_2}$ on the boundary only have non-zero weights for the splits $Q_{j_1},Q_{j_2},\dots, Q_{j_k}$.  
Furthermore, since the projection of each $Q_j$ onto a unique split $P_i$ in $S(T)$ does not depend on the orthant, when we remove the 0 weights from each system of equations ($M_T^{\cO_1} \cdot \mathbf{x}^{\cO_1} = \mathbf{w}$ and $M_T^{\cO_2} \cdot \mathbf{x}^{\cO_2} = \mathbf{w}$), the two systems of equations will now be identical.  
Therefore the intersection of $E^{\cO_1}_T$ and $E^{\cO_2}_T$ is precisely each of their intersections with the boundary orthant $\cO_1 \cap \cO_2$.  
\end{proof}

\begin{exmp} Returning to the tree $T$ from Examples 3.4 and 3.6, based on the projection $\Psi_{\bar{4}}(Q_j)$ which deletes the label ``4", we set up the following linear system.
$$\left\lbrace\begin{array}{l}
x_{(25)(134)}+x_{(13)(245)} = 0.2 = w_{(13)(25)}\\
x_{(24)(135)} + x_{(2)(1345)} = 0.3 = w_{(2)(1345)}\\
x_{(45)(123)}+x_{(5)(1234)} = 0.25 = w_{(5)(1234)}\\
x_{(14)(235)} + x_{(1)(2345)} = 0.15 = w_{(1)(2345)}\\
x_{(34)(125)} + x_{(3)(1245)} = 0.2 = w_{(3)(1245)}\\
x_j \geq 0\hspace{1em} \forall j
\end{array}\right.$$
Without the leaf dimensions, the portion of the extension space pictured in Example 3.4 is specified by the first equation and the non-negative constraints.
\end{exmp}

\begin{theorem}\label{Eispreimage} Let $\leafset \subset [N]$ and $T \in \cT^{\leafset}$.  Then $E_T^N = \Psi_{\leafset}^{-1}(T) \subset \mathcal{T}^{N}$.
\end{theorem}
\begin{proof}
By construction and Proposition (\ref{CT}), $E_T^N \subset S_T^N$, so $\Psi_{\leafset}(S(\bar{T}))= S(T)$ for each $\bar{T}\in E_T^N$, i.e. $E_T^N$ and $\Psi_{\leafset}^{-1}(T)$ intersect the same orthant set, given by $S_T^N$. Furthermore, the procedure of dimension reduction as given in Definition 2.4 guarantees that each edge $e_i \in \cE(\Psi_{\leafset}(\bar{T}))$ will be obtained by concatenating edges $\bar{e_j}$ projecting to $e_i$. Thus, to satisfy $T = \Psi_{\leafset}(\bar{T})$, for a fixed orthant $\cO \in S_T^N$, there is a fixed procedure of dimensionality reduction, and a fixed set of splits $\{Q_j\}$, each with weight $\bar{w_j}$, projecting to some $P_i\in S(T)$. Therefore $\Psi_{\leafset}(\bar{T}) = T$ is equivalent to having $\sum_{j: \Psi_{\leafset}(Q_j)= P_i} \bar{w_j} = w_i$ for each $e_i \in \cE(T)$ with weight $w_i$, which is precisely the condition specified by the equations of $E^{\cO}_T$. Since $E_T^N$ and $\Psi_{\leafset}^{-1}(T)$ agree in each orthant, we have the result.
\end{proof}

%\begin{remark}With minimal pre- and post-processing, this algorithm can be used for more general cases, which will be heavily utilized in the next section. Given a phylogenetic tree $T \in \mathcal{T}^{\leafset(T)}$ with a set of $n$ labels $L(T)\subset \mathbb{N}$, and given $\lambda$ a disjoint set of leaf labels to be added, $|\lambda| = \ell$, we can calculate $E^{L(T)\cup \lambda}_{T}\simeq E_{T,n,\ell}$. Since $\mathcal{T}^{L(T)} \cong \mathcal{T}^n$ via a relabeling (which induces an isometry), and similarly $\mathcal{T}^{\lambda\cup L(T)} \cong \mathcal{T}^{n+\ell}$ via an aligned relabeling $R$ (i.e. $R:\lambda \mapsto [n], \lambda^* \mapsto \{n+1,\dots,n+\ell\}$), we can take the image $R(T) \in \mathcal{T}^n$, calculate $E_{R(T),n,\ell}$, and take $E^{L(T)\cup \lambda}_{T} = R^{-1}(E_{R(T),n,\ell})\subset \mathcal{T}^{L(T)\cup \lambda}$.
%\end{remark}

\hspace{-1.5em}{\bf Complexity Results}\\

If we restrict our computation to a single orthant, the matrix $M^{\cO}_T$ can be computed by calculating each $\Psi_{\leafset}(Q_j)$ and matching with $P_i$, which is $O(N)$. Each computation like this determines a column of $M^{\cO}_T$ (with unique non-zero entry in $i$-th position), so $M^{\cO}_T$ is computed in $O(N^2)$. 

The barrier to a polynomial time algorithm is the size of $C^N_T$, which by \cite{UIUC}, we have is 
$$ \frac{(2(n+\ell)-5)!!}{(2n-5)!!} \in O (N^\ell).$$
These two estimates imply that computing all extension matrices is less than quadratic in the support size of the space. 

\begin{prop} The computation of the collection of matrices $M^{\cO}_T$ is $O(N^{\ell+2})$. Since this dominates the complexity for the previous steps, the total complexity of this algorithm is $O(N^{\ell+2})$.
\end{prop}
\begin{proof}
Combining with Proposition \ref{connection runtime}, the total algorithm will be dominated by $N^{\ell+2} + 2^{3\ell}n^3$, and so we have the complexity bound given in the statement. For $\ell << n$ fixed, this is polynomial of degree $\ell + 2$. 
\end{proof}

The actual space of solutions, a convex affine polytope, can be presented by its boundary vertices in each orthant; interior points can then be expressed as convex combinations of boundary vertices. These can be computed, but there are a lot of them: since $M$ is rank $n$, we expect around ${N \choose n}$ basic feasible solutions, which gives an estimate for boundary vertices. In low dimensions, enumeration might be reasonable; there exist algorithms to do this. In general, we will operate on this space in indirect ways. 

\begin{lemma}\label{decision} 
Let binary tree $T \in \cT^{\leafset}$ with $\leafset \subset [N]$, $|\leafset| = n$, and $|N \backslash \leafset| = \ell$.  To test whether a point $\bar{x} \in \mathcal{T}^N$ is in $E^N_T$, it is sufficient to check whether $\Psi_{\leafset}(\bar{x}) = T$, which is $O(N)$.
\end{lemma}

\begin{proof}
The first part is obvious from Theorem \ref{TDR}. For the complexity, we note that in order to check the latter condition, we must perform dimensionality reduction on $\bar{x}$, which can be done in $O(\ell)$ from the tree representation of $\bar{x}$: each successive leaf removal results in at most one concatenation (see Definition \ref{TDR}). Then we must compare $\Psi_{\leafset}(\bar{x})$ to $T$. Since they are both binary trees in $\mathcal{T}^{\leafset}$, they each have $2n-3$ splits and, as graphs, $2n-4$ vertices. We can therefore determine isometry by a simultaneous traversal, which is $O(n)$. Since $N >n,\ell$, we have the result, which is not tight.
\end{proof} 
For the more general statement of this, see Prop (\ref{Ecomp}).

\begin{remark}
To find a point $\bar{x}$ in $E^{\cO}_T$ which optimizes a linear function $f(\bar{x})$ in orthant $\cO$, standard linear programming methods will find a global solution in polynomial time, with an average runtime $\sim N^3 B$ using the simplex method. To estimate $B$, we note that matrices $M^{\cO}_T$ will always be $2n-3 \times 2N-3$ (binary) matrices, with $2n-3$ float edge lengths, requiring a total of $O(Nn)$ bits, for a total average run time on the order of $N^4n$.
\end{remark}

\subsection{Comparing extension spaces}
One might hope that, as we have $d_{\mathcal{T}^{\leafset}}(\cdot,\cdot)$ which gives a well-defined metric on $\mathcal{T}^{\leafset}$, we might be able to use this metric to define a meaningful distance between $E^N_{T_1}$ and $E^N_{T_2}$ as sets. Though this calculation is possible, distances between the sets $E_1$ and $E_2$ in $\mathcal{T}^N$ do not produce a metric on extension spaces.
%
%Given $E_{T,n,\ell}$, we may now compare this object with a tree $\mathfrak{T} \in \mathcal{T}^{n+\ell}$. Since $E_{T,n,\ell}$ is closed, we have a well-defined distance 
%$$ d(E_{T,n,\ell},\mathfrak{T}) = \inf_{\bar{T} \in E_T}d_{\mathcal{T}}(\bar{T}, \mathfrak{T})$$
%which will be 0 iff $\mathfrak{T} \in E_{T,n,\ell}$.
%We denote by $\mathcal{E}^N$ the space of extension spaces with $n+\ell=N$, including, as extension spaces $E_{\mathfrak{T},N,0}$, points in $\mathcal{T}^N$ without modification. 
%
%%The extension spaces with $\ell \geq 1$ are never compact, but we can also define such a distance on two extension spaces, by the following lemma.
%%\begin{lemma} Let $E_{T,n,\ell}, E_{T',n',\ell'}$, with $N = n+\ell = n'+\ell'$, be extension spaces of $T$ and $T'$, respectively. Then the function
%%$$ d_{\mathcal{E}^N}(E_{T,\ell},E_{T',\ell'} ):=\inf_{\bar{T}\in E_T, \bar{T'}\in E_{T'}}d_\mathcal{T}(\bar{T},\bar{T'})+ $$
%%is  on $\mathcal{E}^N$.
%%\end{lemma}
%%\begin{proof}
%% 
%%\end{proof}
%
\begin{remark} The distance function $d_{E^N}: (E^N_T,E^N_{T'})\mapsto\inf_{\bar{T},\bar{T'} \in \cT^N}d_{\cT^N}(\bar{T},\bar{T'})$ is not a pseudometric. To see this, take two distinct points $\mathfrak{T}_1, \mathfrak{T}_2$ in a non-trivial extension space $E$; they are each trivial extensions of themselves, so they are in the domain of the distance function, and there is a positive tree space distance $d_{\mathcal{T}^N}(\mathfrak{T}_1,\mathfrak{T}_2) = d_{E^N}(\mathfrak{T}_1,\mathfrak{T}_2)$. However, each have $\inf_{\bar{T}\in E}(\mathfrak{T}_i,\bar{T}) = 0$, $i =1,2$, so $\inf_{\bar{T}\in E}(\mathfrak{T}_1,\bar{T}) + \inf_{\bar{T}\in E}(\mathfrak{T}_2,\bar{T}) = 0$, which violates the triangle inequality. Furthermore, $d(E_1,E_2) = 0$ and $d(E_2,E_3) = 0$ do not imply $d(E_1, E_3) = 0$. 
%Furthermore,The function $(E_1,E_2)\mapsto\inf_{\bar{T},\bar{T'}}d_\mathcal{T}(\bar{T},\bar{T'})$ isn't necessarily achieved by a unique geodesic in $\mathcal{T}^{n+\ell}$. For example, take any tree $T$ such that $|e|>1$ for all edges $e$ of $T$, with $E_{T,n,\ell}$, and let $T'=T-1$ be the tree with the same split set as $T$, for which each edge $|e|_{T'} = |e| -1$, i.e. the length of the split in $T'$ is 1 less than the corresponding edge length in $T$. Then $E_{T',n,\ell}$ will be parallel to $T$ in all orthants, with a fixed distance of 1 corresponding to the translation of each preimage tree. 
\end{remark}

However, the vanishing of this quantity is meaningful, and it will correspond to a ``compatibility" of trees:

\begin{lemma} 
Let $E_{T,n,\ell}, E_{T',n',\ell'}$, with $N = n+\ell = n'+\ell'$, be extension spaces of $T$ and $T'$, respectively. Then $ d(E_{T,n,\ell}, E_{T',n',\ell'}) = 0$ if and only if there exists a tree $\mathfrak{T}\in \mathcal{T}^N$ which contains all the splits of $T$ and all the splits of $T'$, with lengths as in $T$ and $T'$.
\end{lemma}
\begin{proof}
If distance is zero then they intersect, since extension spaces are locally affine. If they intersect, their intersection is non-empty, and we can choose a tree $\mathfrak{T}$ in this intersection. Then by Proposition (\ref{Eispreimage}), $\mathfrak{T}$ projects to each of $T$ and $T'$ under $\Psi_{\leafset(T)}$ and $\Psi_{\leafset(T')}$, and so $\mathfrak{T}$ contains a preimage of each split $P \in T, P'\in T'$, which separates the same leaves that $P$ and $P'$ do. Furthermore by previous results we know that the pairwise distances between leaves are preserved between $T$ and $\mathfrak{T}$ (and $T'$ and $\mathfrak{T}$).
\end{proof}
Then $\mathfrak{T}$ can be seen as combining the information of $T$ and $T'$, as in the case that $T$ and $T'$ are samples of a larger tree on different taxa subsets, and this $ d(E_{T,n,\ell}, E_{T',n',\ell'}) = 0$ case (and later, $ d(E_{T,n,\ell}, E_{T',n',\ell'}) <\epsilon$) is what we will explore in the next section.

\section{Extension of tree sets}

 By Theorem \ref{Eispreimage}, an intersection point of two extension spaces is an intersection of the preimages. In particular, if $\bar{T}\in \mathcal{T}^N$ is contained in $\Psi_{\leafset(T)}^{-1}(T)$ and $\Psi_{\leafset(T')}^{-1}(T')$, then by definition, $\Psi_{\leafset}(\bar{T}) = T$ and $\Psi_{\leafset'}(\bar{T}) = T'$. This means that $\bar{T}$ can be seen as ``combining" the information of two ``compatible" trees, with different leaf sets $\leafset$ and $\leafset'$.

\begin{exmp} Building on Example 3.4, suppose we have a second tree $T'$ with labels $\leafset(T') = \{1,2,3,4\}$, leaf edge lengths $(0.15,0.3,0.2,0.35)$ respectively, and interior edge $(13)(24)$ with length 0.15, pictured on the left in Figure~\ref{fig:intersection}. Then the preimage of $T'$, shown in the center of Figure~\ref{fig:intersection}, under pruning of the 5th leaf is also a $T'$-shaped subspace of $\mathcal{T}^5$, and it intersects $\Psi_{\bar{5}}^{-1}(T)$ in a single point (circled), $(0.05,0.15)$ in the $(13)-(25)$ plane (green), representing the tree pictured on the right in Figure~\ref{fig:intersection}, with leaf edges $(0.15,0.25,0.3,0.2,0.35,0.25)$ %$(0.15,0.3,0.2,0.35,0.25)$
, respectively. This combination of information can also be realized as the pairwise path distance matrix of $\bar{T}$, which contains the distance matrices for $T$ and $T'$ as distinct minors.

\begin{figure}
    \centering
    \includegraphics[scale=0.5]{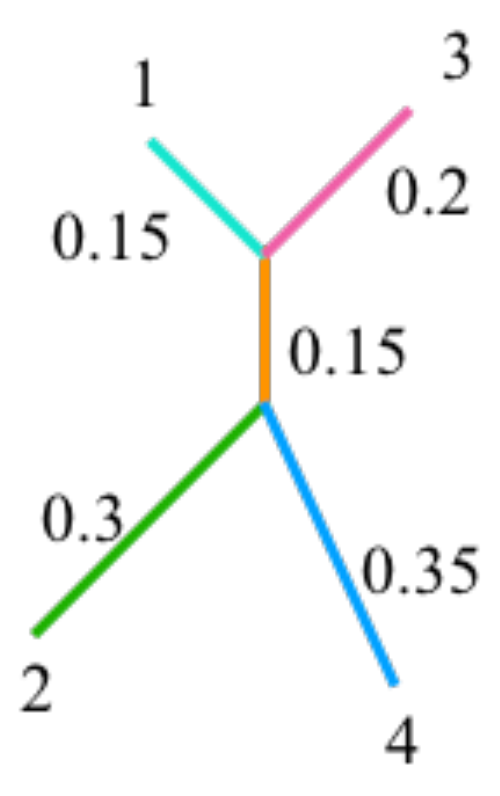}
    \includegraphics[scale=0.5]{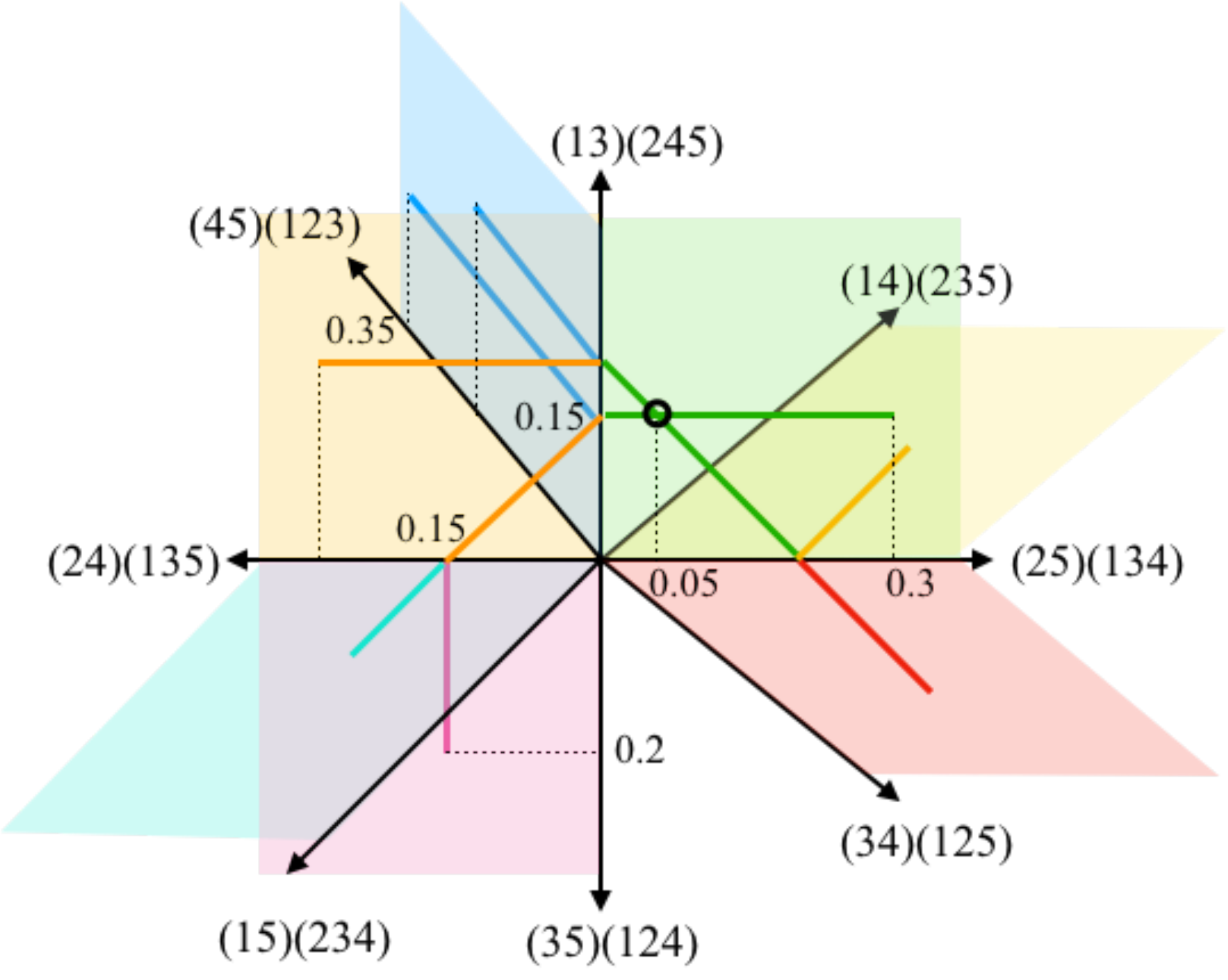}
    \includegraphics[scale=0.5]{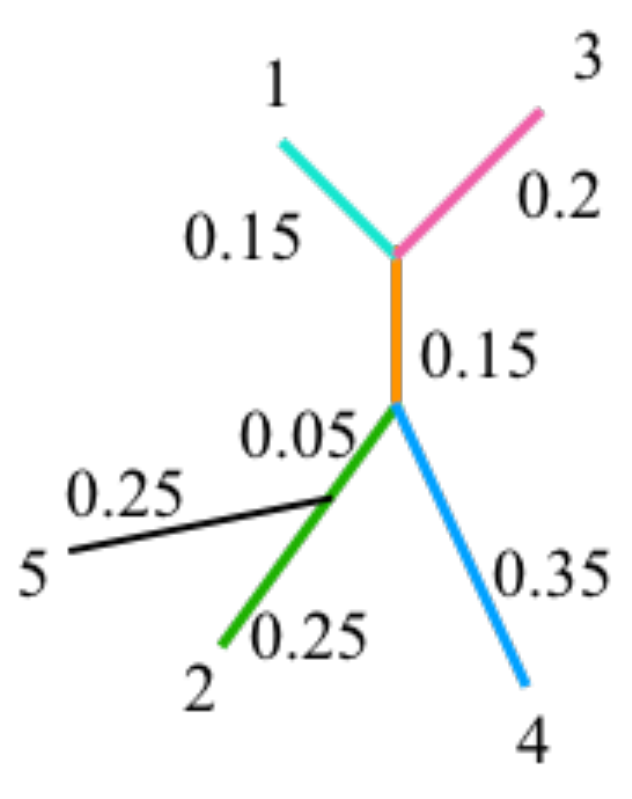}
    \caption{Left, a second tree $T'$ with leaves $\{1,2,3,4\}$. Center, the $T$-shaped subspace of $\Psi_{\bar{5}}^{-1}(T)$ and the $T'$-shaped subspace of $\Psi_{\bar{5}}^{-1}(T')$, with their unique intersection circled.  Right, the tree at the intersection point of the two subspaces. }
    \label{fig:intersection}
\end{figure}

$$ A_{\bar{T}} = \left(
\begin{array}{l l l l l }
0 & .65 & .35 & .65 & .6 \\
.65 & 0 & .7  & .7 & .55\\
.35 & .7 & 0 & .7 & .65 \\
.65 & .7 & .7 & 0 & .65 \\
.6 & .55 & .65 & .65 & 0 
\end{array}\right)$$

\end{exmp}

In this section, we are interested in characterizing intersection points of this type, and quickly computing the equations which define the complete set. More generally, consider a collection of trees $\mathbf{T} = T_1,\dots,T_k$ with leaf-sets $\leafset_r$, where $|\leafset_r| = n_r$.  By fixing $\ell_r = N - n_r$ we consider their tree dimensionality reduction preimages $\Psi_{\leafset_r}^{-1}(T_r)$ collectively in $\mathcal{T}^N$. We can now define generalizations of the 
\begin{itemize}
\item {\em connection cluster} $C^N_\mathbf{T} := \cap_{r} C^N_{T_r}$, 
\item {\em connection space} $S^N_\mathbf{T} := \cap_r S^N_{T_r}$, and 
\item {\em connection graph} $G^N_\mathbf{T} := \cap_r G^N_{T_r}$. 
\end{itemize} 
These correspond to the topologies in $\cT^N$ which simultaneously extend $S(T_r)$ for all $T_r \in \mathbf{T}$. 

As in Section 3, where $\mathbf{T} = \{T\}$, we will proceed by first finding $C^N_\mathbf{T}$, and finding solutions to a system of metric constraints, which will provide us with the {\bf intersection extension space} $E^N_\mathbf{T}:= \cap_r E^N_{T_i}$. 

However, due to the high codimension of $E^N_{T_r}$, the extension space of $\mathbf{T}$ can be unstable under small treespace perturbations of the $T_r$. In the next section, we will present a relaxation which will allow for bounded independent perturbations of $T_1,\dots, T_k$, which produces a neighborhood of each $E^N_{T_r}$ for transverse intersection. This also gives rise to two ``measures of compatibility", $\alpha_\mathbf{T}$ and $p_\mathbf{T}$, the minimum parameter under two relaxation regimes giving a non-empty extension intersection.

In the final section, we will discuss methods for consolidating more diverse tree topologies, which will choose orthants of highest likelihood for analysis.

It may be useful to first give a few remarks on $N$. We are assuming that the data has consistently labeled trees - i.e. that label $j$ represents the same sample across trees in $\mathbf{T}$. %This means that the tree $T_i$ will not live in the standard $\mathcal{T}^{n_i}$, but instead the isometric $\mathcal{T}^{L_i}$. 
If the labels are numbers, we may want to take $N$ equal to the maximum label, to represent missing taxa, but it might also make sense to take $N$ equal to the number of different labels, which would simplify the solution space and decrease computation time, and add degrees of freedom later. Whatever $N$ is chosen, we will assume that the label set $\leafset_r$ of $T_r$ is a subset of $[N]$, and we will denote by $\Psi_{\leafset_r}$ the TDR projection map from $\mathcal{T}^N$ to $\mathcal{T}^{\leafset_r}$.

\subsection{Combinatorial intersection}

Given $T_1,\dots,T_k$ binary trees with leaves $\leafset_r$ such that $L_r \subset [N]$ for each $r$, we can construct $G^N_{T_r}$ for each $r$, and take the intersection, to find tree topologies which project under $\Psi_{\leafset_r}$ to $S(T_r)$ for each $r$.  However, if we are starting from the split sets $S(T_r)$, it is much more efficient to construct the intersection itself, since it can be much smaller than the largest $G^N_{T_r}$. The algorithm proceeds via the following steps.
\begin{enumerate}
\item Reindex the trees so that $T_1$ has the greatest number of leaves $n_1$, and therefore the smallest $\ell_1$. This will ensure that we begin with the smallest connection graph.
\item Generate $G = G^N_{T_1}$. 
\item For each $Q \in V(G)$, check if $\Psi_{\leafset_r}(Q_j)\in S(T_r)$ for all $i = 2,\dots,k$. If not, remove $Q$ from $G$, as well as all of its incident edges. 
\item Find $(2N-3)$-cliques in $G$, output this set as $C^N_{\mathbf{T}}$. 
\end{enumerate}

\begin{prop} Given $\mathbf{T} = \{T_r\}$ a finite set of binary trees, and $N$ such that $\leafset_r \subset [N]$ for each $r$, then $G = \bigcap_{r}G^{N}_{T_r}$, and therefore topology $C \in C^N_{\mathbf{T}}$ if and only if $\Psi_{\leafset_r}(S(C)) = S(T_r)$ for each $T_r$. 
\end{prop}
\begin{proof}
By construction of the final graph $G$, $V(G)$ consists of splits $Q_j$ such that $\Psi_{\leafset_r}(Q_j) \in S(T_r)$ for each $r$. This is the vertex set of $\cap_r G^N_{T_r}$, by construction. The edges of $G$, formed in Step 2, come from pairwise compatibility, which is independent of the original tree set. We know also that compatibility determines adjacency equally for each $G^N_{T_r}$, so that the intersection of connection graphs is the full subgraph of the intersection of the vertex set in $L_N^1$, and any edge which is present in $G^N_{T_1}$ is present in all $G_{T_r}$ containing both endpoints. Therefore all edges of $\cap_r G^N_{T_r}$ are present in Step 2, and none are deleted, since their endpoints remain. So $G = \cap_r G^N_{T_r}$. 

We can also note that if $K$ is a maximal $(2N-3)$-clique in $G$, then $K$ is also a maximal clique in each $G^N_{T_r}$, and conversely, so that $C^N_\mathbf{T}= \cap_r C^N_{T_r}$.

Next, we note that by Proposition 3.6, topology $C \in C^N_{T_r}$ if and only if $\Psi_{\leafset_r}(S(C)) = S(T_r)$. Then since $C^N_\mathbf{T}=\cap_r C^N_{T_r}$, it follows that $C \in C^N_\mathbf{T}$ if and only if $\Psi_{\leafset_r}(N) = S(T_r)$ for each $r$. 
\end{proof}

\begin{defi} We call a set $\mathbf{T} = \{T_r\}$ of binary trees {\bf combinatorially compatible} if $C^N_{\mathbf{T}}\neq \emptyset$.
\end{defi}
This definition relates to edge compatibility (Definition \ref{compatible}), but edge compatibility is not a special case of it.  The requirement that the inputs be binary trees would need to be generalized.

%\begin{lemma} If $\mathbf{T} = \{T_1,\dots,T_k\}$ distinct with $k > 2^{N-2}$, then $C^N_\mathbf{T}$ is empty.
%\end{lemma}
%\begin{proof}
%If two distinct trees $T_i$ and $T_j$ have the same leaf set, then their extension spaces are disjoint, so the full intersection is empty. The bound above is therefore the number of possible non-trivial leaf sets.
%\end{proof}

\begin{prop}
If $N\cdot k < 2^{2\ell_1}$, then this algorithm is $O(2^{3\ell_1}n_1^3)$. If $N\cdot k > 2^{2\ell_1},$ then it is $O(2^{\ell_1}n_1^4k^2)$. Either way, it is $O(2^{3\ell_1}n_1^4k^2)$. 
\end{prop}
\begin{proof}
Reindexing the trees to put the tree with the most leaves first is $O(k)$. By Proposition 3.8, we have that Step 2 is $O(2^{2\ell_1}n_1^3)$. For Step 3, we iterate through each of $\sim 2^{\ell_1}n_1$ vertices, and for each, delete leaves to get down to $L_r$ (order $N$) and compare with the $2n_r - 3$ splits of $T_r$ (order $n_r(2n_r-3) \sim 2n_r^2 \lessapprox 2n_1^2$).  In total, then, Step 3 is $O(2^{\ell_1}n_1^3Nk)$, and we can simplify to $O(2^{\ell_1}n_1^4k^2)$ by noting that $N < k\cdot n_1$. For Step 4, in the worst case, the size of $G$ is comparable to $G_{T_1}^N$, so by Proposition 3.8, Step 4 is $O(2^{3\ell_1}n_1^3)$. 
If $N\cdot k < 2^{2\ell_1}$, then Step 4 dominates. If not, Step 3 does. 
\end{proof}

\subsection{Metric intersection}
Given a binary topology $C \in C_{\mathbf{T}}^N$%$K_{N-3}$ clique 
with splits $Q_1,\dots, Q_{N-3}$, plus leaf splits $Q_{N-2},\dots,Q_{2N-3}$, we have an $2\ell_r$-dimensional solution space for each $T_r$, cut out by a set of equations $$ x_{m_1}+x_{m_2}+\dots+x_{m_{a_j}} = w_i$$ for each $P_i \in S(T_r)$, $i = 1,\dots, 2n_r - 3$. We take the collection of equations from all $T_r$, and this defines a solution space: either it is empty, or there is some linear subspace of solutions, with dimension at most $\min_r \ell_r$, which simultaneously satisfies the collection of metric constraints. Unlike the single-tree extension case, this system can be overdetermined, and have no solution in an orthant $\cO \in S_\mathbf{T}^N$.

\begin{defi} 
Let $\cO \in S^N_\mathbf{T}$ be an orthant in the intersection cluster, with split lengths parametrized by respective coordinates $(x_1, \dots, x_{2N-3})$.%$(x_1,\dots,x_m,\dots,x_{2N-3})$. 
Let $M^{\cO}_{T_r}$ be the $(2n_r-3)\times(2N-3)$ projection matrix of $S$ to $\mathcal{T}^{\leafset_r}$. We then write 
\begin{equation}  \label{exeq}
\left(
\begin{array}{c}
M^{\cO}_{T_1}\\
\hline
M^{\cO}_{T_2}\\
\hline
\vdots \\
M^{\cO}_{T_k}
\end{array}
\right) \mathbf{x}^{\cO} = \left(
\begin{array}{c}
\mathbf{w}_1\\
\mathbf{w}_2\\
\vdots \\
\mathbf{w}_k
\end{array}
\right), \hspace{1em} \mathbf{x}^{\cO} \geq 0
\end{equation}

Then the solution space of $\mathbf{x}^{\cO}$ satisfying these is denoted $E^{\cO}_\mathbf{T}$. The matrix on the left is denoted $M^{\cO}_\mathbf{T}$ for brevity, and the vector on the right hand side $\mathbf{w_T}$, so expressing the equation more compactly, $M^{\cO}_\mathbf{T}\mathbf{x}_{\cO} = \mathbf{w_T}$. The {\bf intersection extension space} of a collection $\mathbf{T}$ of trees is defined to be $$E_{\mathbf{T}} := \bigcup_{\cO \in S_\mathbf{T}} E^{\cO}_{\mathbf{T}},$$ where as before, $N$ is taken to be the size of the total leaf set $L(\mathbf{T})$ and $\ell_r = N-n_r$ for $T_r \in \mathbf{T}$ of size $n_r$.
\end{defi} 

Note that when $\mathbf{T}= \{T\}$, $E_\mathbf{T}= T$, since $N$ is set to $\leafset(T)$, unless we set a larger extension space, in which case $E^N_\mathbf{T}= E^N_T$, and so the results of Section 3 are a special case of this definition and algorithm.

\begin{defi} Given a set of trees $\mathbf{T}$ as above, we call the set {\bf compatible} if $E_\mathbf{T}\neq \emptyset$.
\end{defi} 
Trivially, for $T\in \mathcal{T}^N$, $\Psi_{\leafset}(T)$ and $\Psi_{\leafset'}(T)$ are compatible for $\leafset,\leafset' \subset [N]$.
%To see this, take two splits $P$ and $Q$ on a common leaf set $n$, and we can view them alternately as degenerate tree topologies in $\mathcal{T}^n$ with one internal edge each. Then they are compa

\begin{prop}\label{intext} For a collection $\mathbf{T}$ of trees with total leaf set of size $N$, the intersection extension region of $\mathbf{T}$ is the intersection of the extension regions of $T \in \mathbf{T}$.  That is, $E^{\cO}_\mathbf{T} = \bigcap_{T \in \mathbf{T}} E^{\cO}_{T}$, $E_\mathbf{T}^N = \bigcap_{T \in \mathbf{T}} E_T^N$.
\end{prop}

\begin{proof}From Proposition 5.2, we know that the orthant support of the intersection is the intersection of the orthant supports. From this, we note that $$\bigcap_T E_T = \bigcap_T \bigcup_{\cO} E_T^{\cO} = \bigcup_{\cO} \bigcap_T E_T^{\cO} = \bigcup_{\cO} E_\mathbf{T}^{\cO},$$
where the first equality is by definition, the middle from finiteness of this union and intersection, and the last equality follows from the fact that the intersection of real linear varieties is the vanishing set of the collection of generating equations.
\end{proof}

\hspace{-1.5em}{\bf Complexity}\\

As in Section 3, we can quickly do the operations that size allows. For $\mathcal{C} = \max \{ \sum_{T_r \in \mathbf{T}} 2n_r - 3, N\}$, equation (\ref{exeq}) is a $\mathcal{C}$-dimensional system of equations which can be set up in $O(kN^2)$. As before, this solution space is cumbersome to describe enumeratively and quick to search. 

\begin{prop}\label{Ecomp} Given $\mathbf{T} = \{T_r\}_{r=1,\dots,k}$, $[N] = \cup L_r$, and a tree $T \in \mathcal{T}^N$, the decision problem ``Is $T$ in $E^N_\mathbf{T}$?" can be solved in $O(kN)$.
\end{prop}
\begin{proof}
To answer the decision problem, it suffices to check, for each $T_r \in \mathbf{T}$, if $\Psi_{L_r}(T) = T_r$. By Lemma \ref{decision}, each can be done in $O(N)$, so the problem is $O(kN)$.
\end{proof}

$C^N_\mathbf{T}$ may be substantially smaller than $C^N_{T_1} \sim N^{\ell_1}$, which may make a complete description possible. A starting point is linear feasibility, i.e. determining if the system (\ref{exeq}) has a solution, which, in contrast to the single-tree case, is not automatically true.  To solve, we introduce $\mathcal{C}$ slack variables $\mathbf{y}_P$ and a $\ell_\infty$-norm variable $\alpha$, and we minimize $\alpha$ subject to
\begin{equation}\label{existence}\begin{array}{r}
\left(\begin{array}{l l l}
M_\mathbf{T} & I
\end{array}  \right)\left(\begin{array}{c}
\mathbf{x}^{\cO}\\
\hline
\mathbf{y}_P\\
\end{array} \right)= \left(\begin{array}{c}
\mathbf{w_T}
\end{array} \right)\\
 x_{m_a}\geq 0 \\
\alpha \geq y_{P} \geq 0
\end{array}
\end{equation} 
This LP has an initial feasible solution: $\mathbf{x}^{\cO} = \mathbf{0}$, $\mathbf{y}_P = \mathbf{w_T}$, and $\min \alpha = 0$ if and only if there is an $\mathbf{x}^{\cO}$ satisfying (\ref{exeq}). This takes as long as your favorite LP solver, for example the simplex method, which will have an average runtime of $O(\mathcal{C}^5)$. In the next section we will investigate the case $\min \alpha > 0$. For the LP formulation, skip to Section \ref{alphaLP}.

\section{Relaxation}

Since each $E^N_{T_r}$ (for collection $T_r$ as in previous section with fixed $N = n_r + \ell_r$) is locally a submanifold of codimension $2n_r-3$ in each orthant, for $n_r + n_{r'} > N+1$, two extension manifolds will not intersect stably. As a result of this, a small perturbation in two different projections of an $N$-tree may give the impression of subtree incompatibility. In the language of our linear optimization problem (\ref{existence}), given a small amount of sampling error in compatible trees, we may obtain an approximate solution with small, but positive, objective value. To resolve this and ensure stability of intersection, we find a minimum amount of error $\alpha_\mathbf{T}$, and find intersections of $\alpha_\mathbf{T}$-neighborhoods of the $E_{T_r}^N$ in each orthant.

\subsection{Uniform $\alpha$-relaxation}\label{alpha}
We can uniformly expand a single orthant of extension region $E_T^{\cO}$ by replacing each equation of the form
$$ x_{m_1}+x_{m_2}+\dots + x_{m_{a_j}} = w_i$$ 
with a pair of equations of the form
$$ x_{m_1}+x_{m_2}+\dots + x_{m_{a_j}} \geq w_i - \alpha $$
$$ x_{m_1}+x_{m_2}+\dots + x_{m_{a_j}} \leq w_i + \alpha $$
Formally, we expand the equation (\ref{exeq}) to the set of inequalities
\begin{equation}  \label{eqalpha}
\left(
\begin{array}{c}
M^{\cO}_{T_1}\\
\hline
M^{\cO}_{T_2}\\
\hline
\vdots \\
M^{\cO}_{T_k}
\end{array}
\right) \mathbf{x}^{\cO} \geq \left(
\begin{array}{c}
\mathbf{w}_1\\
\mathbf{w}_2\\
\vdots \\
\mathbf{w}_k
\end{array}
\right) - \alpha\cdot\mathbf{1}, \hspace{1em}\left(
\begin{array}{c}
M^{\cO}_{T_1}\\
\hline
M^{\cO}_{T_2}\\
\hline
\vdots \\
M^{\cO}_{T_k}
\end{array}
\right) \mathbf{x}^{\cO} \leq \left(
\begin{array}{c}
\mathbf{w}_1\\
\mathbf{w}_2\\
\vdots \\
\mathbf{w}_k
\end{array}
\right) + \alpha\cdot\mathbf{1}, \hspace{1em} \mathbf{x}^{\cO} \geq 0
\end{equation}
For a single tree $T_r$, the solution space in a fixed orthant $\cO$ is the extension space of a rectangular $\alpha$-neighborhood of $T_r$ in $\mathcal{T}^{\leafset_r}$, and we will see that it contains a neighborhood of the $2\ell_r$-plane $E^{\cO}_{T}$ in $\mathcal{T}^N$. When $\alpha <w_i$ for all $P_i \in S(T)$, the solution space does not contain the cone point. The orthant solution space for $\mathbf{T}$ then becomes a (bounded or unbounded, empty or non-empty) polytope $E^{\cO}_\mathbf{T}(\alpha)$. We choose $\alpha$ uniformly across orthants to ensure that the extension polytope is closed for small $\alpha$.

\begin{defi}For a given tree $T\in \mathcal{T}^{\leafset}$, we denote by $E^N_{T}(\alpha) := \bigcup_{\cO \in S_T^N} E^{\cO}_{T}(\alpha)$, and this is called the {\bf $\alpha$-extension region of $T$ in $\mathcal{T}^N$}. \end{defi}

\begin{exmp} Let $\alpha = 0.05$, then the $\alpha$-extension region of our first example is shown in Figure~\ref{fig:orthants_alpha}.

\begin{figure}
    \centering
    \includegraphics[scale=0.5]{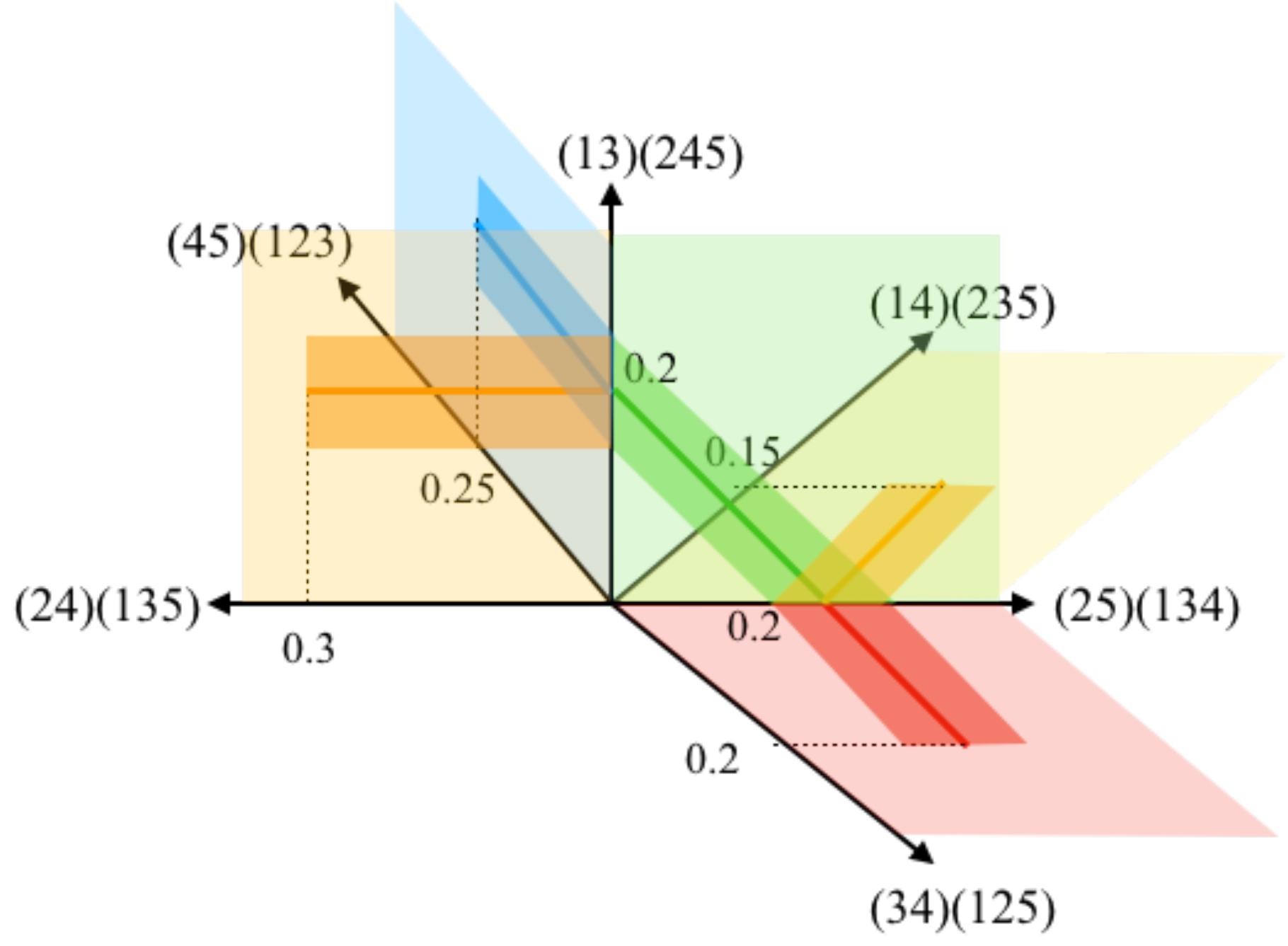}
    \caption{The $\alpha$-extension region of tree $T$ from Example~\ref{ex:running_start} is the darker shaded region within the 5 orthants.  Here $\alpha = 0.05$. }
    \label{fig:orthants_alpha}
\end{figure}

\end{exmp}
\begin{defi} For a finite collection $\mathbf{T} = \{T_r\}$ of binary trees and orthant $\cO \in S^N_\mathbf{T}$, the $\alpha$-relaxation of the equations (\ref{exeq}) gives a (possibly empty) polytope in $\cO$, denoted $E^{\cO}_\mathbf{T}$. The {\bf $\alpha$-intersection region} of $\mathbf{T}$ is defined to be $$E_{\mathbf{T}}(\alpha) := \bigcup_{\cO \in S_\mathbf{T}^N} E^{\cO}_{\mathbf{T}}(\alpha),$$ where as before, $N$ is taken to be the size of the total leaf set $\cup \leafset_r$ and $\ell_r = N-n_r$.
\end{defi}

\begin{prop} Let binary tree $T \in \cT^{\leafset}$ have leaf-set $\leafset \subset [N]$.  If tree $\mathfrak{T} \in E_\mathbf{T}^{\cO}(\alpha)$, then $d_\mathcal{T^{\leafset}}(\Psi_{\leafset}(\mathfrak{T}),T) < c\alpha$ for all $T\in \mathbf{T}$, where $c$ is a constant depending on $\cO$ and $\leafset(T)$.
\end{prop}

\begin{proof}
If $\mathfrak{T}\in E^{\cO}_\mathbf{T}(\alpha)$, then there is some $\mathfrak{T}' \in E^{\cO}_{\mathbf{T}}$ such that $d(\mathfrak{T},\mathfrak{T}')< \alpha$. Since $\mathfrak{T}'\in E_\mathbf{T}^{\cO}$, we have that $\Psi_{\leafset}(\mathfrak{T}') = T$ for all $T\in \mathbf{T}$. So by 
Section 4.3 in \cite{ZKBR2}, we can take $c = \log_2(N)$ to be the max number of edges concatenated in $\Psi_{\leafset}$ acting on $S(\mathcal{T}^N)$.
 \end{proof} 

Note that $E_T^N(\alpha)$ is not defined as an $\alpha$-neighborhood of $E_{T}^N$, but its restriction to each orthant in $S_T^N$ is an $\alpha$-neighborhood in that orthant. Furthermore, for small $\alpha$, $E_T^N(\alpha)$ is closely related to the neighborhood. 

\begin{prop} \label{neighborhood}
Let $T \in \cT^{\leafset}$ be a binary tree with leaf-set $\leafset \subset [N]$.
For $\alpha < \log_2(N)^{-1}\min_{e\in \cE(T)} w_e$, $E_T^N(\alpha)$ contains the $\alpha$-neighborhood of $E_{T}^{N}$ in $\mathcal{T}^N$.
\end{prop}
\begin{proof}
The $\alpha$-neighborhood $N_{\alpha} := N_\alpha(E_T^N) \subset \mathcal{T}^N$ is path-connected. Suppose $\mathfrak{T}\in N_\alpha \backslash E_T^N(\alpha)$. Since $N_\alpha \cap \cO = E_T^N(\alpha)$ for $\cO \subset S_T^N$, we conclude that $\mathfrak{T} \notin \cO$ for any orthant of the connection space, so the orthant $\cO'$ containing $\mathfrak{T}$ does not contain a preimage of some edge $e'\in \cE(T)$, i.e. $e' \notin\Psi_{\leafset}(S(\cO))$. Since the neighborhood is path-connected though, between $\mathfrak{T}$ and $E_T^N$ there is some geodesic path $\gamma$ contained in $N_\alpha$, corresponding to a deformation of $\mathfrak{T}$ to some tree $\bar{T}\in E_T^N$. 

Consider the image of $\gamma$ under $\Psi_{\leafset}$. $\Psi_{\leafset}(\mathfrak{T})$ does not have edge $e'$, so $\Psi_{\leafset}(\gamma)$ must have length at least the length of the projection to $e'$.  Therefore the length of $\Psi_{\leafset}(\gamma)$ must be greater than $\alpha$ since the $e'$ component of the path has length at least $w_{e'}\geq\min_e w_e >\log_2(N)\alpha$. Since by \cite{ZKBR} geodesic lengths grow by at most $\log_2(N)$ under $\Psi_{\leafset}$, this implies that $\mathfrak{T}\notin N_\alpha$, a contradiction.
\end{proof}

\begin{lemma}
\label{alphastab} 
Let $\mathbf{T} = \{T_1, ..., T_r\}$ be a set of binary trees in $\cT^N$, each with leafset $\leafset(T_r) \subset [N]$.  If $\alpha_1 < \alpha_2$, then $E_\mathbf{T}^N(\alpha_1) \subset E_\mathbf{T}^N(\alpha_2)$. For $T,T'\in \mathcal{T}^{\leafset}$ with $d_{\mathcal{T}^{\leafset}}(T,T')< \min\{\alpha,\log_2(N)^{-1}\min_{e_j \in T} w_j\}$, we have the inclusion $E_{T'}^N \subset E_{T}^N(\alpha)$.
\end{lemma}

\begin{proof}
The first statement is clear from construction. For the second, if $d_{\cT^{\leafset}}(T,T')< \min w_j/\log(N)$, then we have that $T'$ has the same split set as $T$, with $w_j'$ the corresponding lengths. Each $w_j'<w_j +d_{\cT^{\leafset}}(T,T') <w_j + \alpha$, and similarly $w_j'>w_j -d_{\cT^{\leafset}}(T,T') > w_j - \alpha$, so solutions to $x_{m_1}+x_{m_2}+\dots+x_{m_{a_j}} = w_j'$ satisfy both inequalities.
\end{proof}

\begin{defi} For a combinatorially compatible collection $\mathbf{T}$ of trees $T_i$ as above, and a given orthant $\cO \in S_\mathbf{T}^N$, we denote by $\alpha^{\cO}_{\mathbf{T}}$ the infimum of $\alpha$ such that $E^{\cO}_{\mathbf{T}}(\alpha)$ is non-empty.
 Then the {\bf intersection parameter} $\alpha_{\mathbf{T}} := \min_{\cO} \alpha^{\cO}_{\mathbf{T}}$. 
\end{defi}

If $\mathbf{T}$ can be obtained from a single $N$-tree by deleting subsets of the leaves, then $\alpha_{\mathbf{T}} = 0$. We also have a natural upper bound on $\alpha^{\cO}_{\mathbf{T}}$ given by the length of the longest edge in $\mathbf{T}$ (so that $E_{\mathbf{T}}(\alpha)$ contains all $E_T(\alpha)$), so $\alpha_{\mathbf{T}}^{\cO}$ is guaranteed to be finite. The parameter $\alpha_\mathbf{T}$ represents minimum amount the preimages of the trees $T_r$ must be perturbed to have a metric solution, assuming combinatorial compatibility.

\subsubsection{Computing $\alpha_\mathbf{T}$}
\label{alphaLP}

When the system of equations (\ref{existence}) has a non-zero optimal solution, we conclude that (\ref{exeq}) had no solutions in that orthant, but we also obtain a valuable by-product: a measure of the degree to which the extension spaces $E^N_{T_r}$ miss each other. For a solution $\mathbf{x}^{\cO}, \mathbf{y}_P$ to (\ref{existence}), for each $r = 1,\dots, k$ we have a unique subset $(\mathbf{y}_P)_r \subset \mathbf{y}_P$, satisfying only the system of equations corresponding to the $M_{T_r}^{\cO}$ rows of $M_\mathbf{T}^{\cO}$. Rearranging those rows,
\begin{equation}\label{exeqr}
(\mathbf{y}_P)_r = \mathbf{w}_r- M_{T_r}^{\cO} \mathbf{x}^{\cO} 
\end{equation}
Thus the $(\mathbf{y}_P)_r$ can be viewed as representing the edge lengths of a positive ``error tree" in orthant $\cO$ of $\mathcal{T}^{\leafset_r}$, and the maximum entry in $(\mathbf{y}_P)_r$ is the minimum amount of $\ell_\infty$ error between $T_r$ and a tree satisfying the $T_r$ rows of equation (\ref{exeq}). Then a global solution is the minimum $\ell_\infty$ error which must be tolerated to include all $T_r \in \mathbf{T}$.

To make this precise, we must add another relaxation variable to stretch $E^N_{T_i}$ to include larger trees as well as smaller ones. 

\begin{prop} The uniform relaxation parameter $\alpha^{\cO}_\mathbf{T}$ of a tree set $\mathbf{T}$ in orthant $\cO \in S^N_\mathbf{T}$ is equal to the objective value of the linear program \begin{equation}\label{alphaT}\begin{array}{l l} 
\mbox{minimize }  & \alpha  \\
\mbox{s.t.} & 
\left(\begin{array}{l l l}
M_\mathbf{T}^{\cO} & I & -I
\end{array}  \right)\left(\begin{array}{c}
\mathbf{x}^{\cO}\\
\hline
\mathbf{y}_P\\
\hline
\mathbf{y}_N\\
\end{array} \right)= \left(\begin{array}{c}
\mathbf{w_T}
\end{array} \right)\\
 & 0 \leq x_{m_a} \\
 & 0 \leq y_{P,m}, y_{N,m} \leq \alpha
\end{array}
\end{equation}
\end{prop} 

 To use the intrinsic $\BHV$ metric, which is piecewise $\ell_2$, we could use the objective function $\min \sum y_{P,m}^2 + \sum y_{N,m}^2$, or in order to preserve linearity of the objective function, we can use the $\ell_1$ metric in tree space, minimizing $\sum y_{P,m} + \sum y_{N,m}$. %\megan{What's $y_m$?  $y_{P,m} - y_{N,m}$?}\gill{ I meant a sum over all the $y$'s, we could say $\sum y_{P,m}^2 + \sum y_{N,m}^2$ if that's more clear.}\megan{Can you change it?  I'm still not quite sure how that fits in with both $\min \sum y_m^2$ and $\sum y_m$.}\gill{I think I get where it's intuitively confusing - $y_P$ and $y_M$ are positive slack variables, so in general $y_P + y_M$ (or the squared version) won't have much meaning, but if the sum is minimized, one of $y_P$ or $y_M$ will be zero (for a given $x,y_P,y_N$, if both $y_P$ and $y_N$ are positive, then replacing the larger, say $y_P$ WLOG, with $y_P-y_N$, set $y_N$ to 0, will give the same solution with a smaller objective value), and it will represent the distance in that dimension from $E_\mathbf{T}$. Do you have a good way to say this?} \megan{I actually just meant to change the first line as you did.  For some reason I wasn't seeing it at the time.} To see that, for $T'$ the tree represented by $x^\cO$ in some solution $x^\cO,y_P,y_M$, we have $\min \sum y_{P,m}^2 + \sum y_{N,m}^2 = \ell_2(T',E^N_\mathbf{T})$, we observe that 

Regarding complexity considerations, if $\mathcal{C} = \max \{ \sum_{T_r \in \mathbf{T}} 2n_r - 3, N\}$, then each matrix is size $\sim\mathcal{C}^2$, so the simplex algorithm will run in $(\mathcal{C}^5)$ time on average, although this is emphatically not a worst-case estimate. This will solve $\alpha_\mathbf{T}$, but again we may not want to enumerate the boundary points.

\subsubsection{Computing $E_\mathbf{T}(\alpha)$}
Using $M_\mathbf{T}^{\cO}$, $\mathbf{w_T}$, $\mathbf{x}^{\cO}$, $\mathbf{y}_P$, $\mathbf{y}_N$ as defined previously, $\cO \in S^N_\mathbf{T}$ and $\alpha \geq \alpha^S_\mathbf{T}$, the $\alpha$-relaxed extension space of $\mathbf{T}$ is defined by the equation
\begin{equation}\label{relax}\begin{array}{r}
\left(\begin{array}{l l l}
M^{\cO}_\mathbf{T} & I & 0 \\
M^{\cO}_\mathbf{T} & 0 & -I
\end{array}  \right)\left(\begin{array}{c}
\mathbf{x}^{\cO}\\
\hline
\mathbf{y}_P\\
\hline
\mathbf{y}_N\\
\end{array} \right)= \left(\begin{array}{c}
\mathbf{w_T} + \alpha\\
\mathbf{w_T} - \alpha
\end{array} \right)\\
x_{m_a},y_{m,P}, y_{m,N} \geq 0
\end{array}.
\end{equation}
We can use this description to search $E_\mathbf{T}^N(\alpha)$ for optimal solutions to a linear function (i.e. a function on $\mathcal{T}^N$ whose restriction to orthants is linear, or a linear function supported in a limited number of orthants).

\subsection{Proportional relaxation}
The $\alpha$-extension region, which is closely related to the $\alpha$ neighborhood of $E_\mathbf{T}$ for small $\alpha$ (Proposition \ref{neighborhood}), is a natural choice for relaxation, but we can also choose a neighborhood proportional to the extension region by solving the inequalities

\begin{equation}  \label{eqpalpha}
M^{\cO}_\mathbf{T} \mathbf{x}^{\cO} \geq (1-p_\alpha)\mathbf{w_T}, \hspace{1em}
M^{\cO}_\mathbf{T}
 \mathbf{x}^{\cO} \leq (1+p_\alpha)\mathbf{w_T}, \hspace{1em} \mathbf{x}^{\cO} \geq 0
\end{equation}

 \begin{defi} Let $\mathbf{T}= \{T_r\}$ be a finite set of binary trees, $\leafset_r \subset [N]$, $C^N_\mathbf{T}$ nonempty, and let $\cO \in S^N_\mathbf{T}$. Then for a fixed $p_\alpha \in [0,1]$, the non-negative solutions to (\ref{eqpalpha}) in $\mathbb{R}^N_{\geq 0}$ give a $(2N-3)$-dimensional solution space in $\cO$; the polytope generated with such a $p_\alpha$ is denoted $ E^{\cO}_\mathbf{T}(p_\alpha)_p$, with corresponding {\bf $(p_\alpha)$-proportional extension region} $$E_\mathbf{T}^N(p_\alpha)_p = \bigcup_{\cO \in S^N_\mathbf{T}} E^{\cO}_\mathbf{T}(p_\alpha)_p.$$ 
Then define the {\bf proportional intersection parameter}
$$p_\mathbf{T} = \inf_{E_\mathbf{T}^N(p_\alpha)_p\neq \emptyset} p_\alpha$$
 \end{defi}

 \begin{prop} The proportional intersection parameter $p_\mathbf{T}\in [0,1]$. For each $\cO \in S^N_\mathbf{T}$, set
$$p^{\cO}_\mathbf{T} := \inf_{E^{\cO}_\mathbf{T}(p_\alpha)_p\neq \emptyset}p_\alpha.$$ Then for $p_\alpha < 1$, $p_{\mathbf{T}} =  \min_{\cO} p^{\cO}_{\mathbf{T}}$.
 \end{prop}
 
 \begin{proof}
 For $p_\alpha< 0,$ $1-p_\alpha > 1+p_\alpha$, so the system (\ref{eqpalpha}) has no solutions. This means that $E^{\cO}_\mathbf{T}(p_\alpha)_p = \emptyset$ for all $\cO$. Therefore $p_\mathbf{T}^{\cO} \geq 0$. 
 
 For $p_\alpha >1$, $1-p_\alpha < 0$, so $\mathbf{x}^{\cO} \cong 0$ is a solution to (\ref{eqpalpha}). Since this point is identified in each orthant, $E_\mathbf{T}^N(p_\alpha)_p$ is formally non-empty. This implies that $p_\mathbf{T}\leq 1$, and for $p_\mathbf{T}<1$, the cone point is not in $E_\mathbf{T}^N(p_\alpha)_p$. In this case, since $E_\mathbf{T}^N(\cdot)_p = \cup_{\cO} E_\mathbf{T}^{\cO}(\cdot)_p$, $E_\mathbf{T}^N(\cdot)_p$ is nonempty precisely when one of $E_\mathbf{T}^{\cO}(\cdot)_p$ is non-empty, which occurs at $\min_{\cO} p^{\cO}_\alpha$, showing equality with $p_\mathbf{T}$. 
 \end{proof}
 Note that as with the uniform parameter $\alpha$, the $p=0$ case gives the original extension regions, but unlike the $\alpha$ case, $p$ has a maximum, 1, which includes boundaries of each orthant, including the cone point. This means that we are guaranteed non-empty relaxed intersection extension region for some value of $p$. Also, for $\alpha< \frac{p}{\log_2(N)}\cdot \min_{e\in \mathbf{T}}w_e$, by Proposition (\ref{neighborhood}) $N_\alpha \subset E_\mathbf{T}^N(\alpha) \subset E_\mathbf{T}^N(p)_p$.
 
 We are also led to a slightly different notion of stability, or alternately, the condition on the following lemma can be strengthened to $d_{\cT^{\leafset}}(T,T')< \min_{e\in \cE(T)} p\cdot w_e$ to obtain the same inclusion.

 \begin{lemma} For any $N \in \NN$ with leafset $\leafset \subset N$, let $T,T'\in \cO \in \mathcal{T}^{\leafset}$, and let $p_\alpha \in [0,1)$. If $|w_e-w'_{e}| < p_\alpha w_e$ for each $e \in \cE(T)$, then $E^N_{T'} \subset E^N_T(p_\alpha)_p$ for any extension codomain $\mathcal{T}^N$. %\megan{what do you mean by extension codomain?  Can we just leave off "for any ...."?}\gill{I mean target tree space for extension, we could say for any $N$ with $\mathcal{L}\subset [N]$ or something}
\end{lemma}
\begin{proof}
Similar to (\ref{alphastab}) in the previous section, we can easily see that solutions to equations for $E_{T'}^N$ satisfy the inequalities defining $E_T^N(p_\alpha)$.
\end{proof}
 
 \begin{prop} The proportional relaxation parameter $(p_\alpha)^{\cO}_\mathbf{T}$ of a tree set $\mathbf{T}$ in orthant $\cO \in S^N_\mathbf{T}$ is equal to the objective value of the linear program \begin{equation}\label{palphaT}\begin{array}{l l} 
\mbox{minimize }  & p  \\
\mbox{s.t.} & 
\left(\begin{array}{l l l}
M_\mathbf{T}^{\cO} & I & -I
\end{array}  \right)\left(\begin{array}{c}
\mathbf{x}^{\cO}\\
\hline
\mathbf{y}_P\\
\hline
\mathbf{y}_N\\
\end{array} \right) =  \left(\begin{array}{c}
\mathbf{w_T}
\end{array} \right) \\
 & 0 \leq x_{m_a}, y_{P,m}, y_{N,m} \\
 & 0 \leq p\cdot w_m - y_{P,m} \\
 & 0 \leq p\cdot w_m - y_{N,m} 
\end{array}
\end{equation}
\end{prop} 

\section{Future work}\label{FW}

Based on the quantities $\alpha_\mathbf{T}$ and $(p_\alpha)_\mathbf{T}$ we have just defined, some future directions can be outlined. 
\begin{enumerate}
    \item As we saw earlier, $\alpha$ and $p_\alpha$ do not compute a distance on extension spaces, but rather a degree of (in)compatibility, with zero values on compatible tree sets. These parameters could potentially be used to produce a quantitative ``test of compatibility" in which some threshold value is set, and the parameter is used to determine whether or not a set of trees is sampled from mutually compatible phylogenies. Unfortunately, although $p_\alpha$ is normalized to lie in $[0,1]$, its value does not admit a uniform statistical interpretation.
    \item Another direction of work is an expansion of results such as Proposition \ref{neighborhood} and Lemma \ref{alphastab} to arbitrary distances, which would require a means of including orthants adjacent to those in the support of $E_\mathbf{T}$. Of particular interest might be a relaxation of the calculation of $C^N_\mathbf{T}$ to include orthants $\mathcal{O}$ satisfying $\Psi_{\mathcal{L}_r}(\mathcal{O}) = S(T_r)$ for a majority of $T_r \in \mathbf{T}$, rather than all. 
    \item Finally, we might extend the idea of extension spaces, their relaxations, and their intersections to tree sets including unresolved (non-binary) trees.  In this case, the  edge lengths near unresolved nodes may have an alternative mathematical interpretation. For example, if multiple leaves attach to a node, then the edge lengths to each of those leaves may represent the approximate path length to those leaves after resolving the node. The results of Section 3 can be expected to follow in a straightforward manner for the unresolved case, but Sections 4 and 5 present some challenge. 
 \end{enumerate}

\section*{Acknowledgements}
We thank Ruth Davidson, Aasa Feragen, Andy Ma, Ezra Miller, Amy Willis, and Th\'{e}r\`{e}se Wu for helpful discussions. The first author is especially grateful to her advisor, Andrew Blumberg, for advice and editing. This work was partially supported by a grant from the Simons Foundation (\#355824, Megan Owen),
NIH grants 5U54CA193313 and GG010211-R01-HIV, and AFOSR grant FA9550-15-1-0302.

\end{document}